\documentclass{siamart1116}
\usepackage{amsmath}
\usepackage{amssymb}
\usepackage{geometry}
\usepackage{graphicx} 
\usepackage{url}
\usepackage{bm}
\usepackage{dcolumn}
\usepackage{color} % only need this for edits
\usepackage{lineno}

\newcommand{\R}{\mathbb{R}}
\newcommand{\Vn}{\mathcal{V}_N}
\newcommand{\Vd}{\mathcal{V}_D}

\newcommand{\ub}[1]{^{(#1)}}
\newcommand{\V}{\mathcal{V}}
\newcommand{\E}{\mathcal{E}}
\newcommand{\etaeff}{\eta_{\textrm{eff}}}
\usepackage{tabularx}
\usepackage{tabulary}
\usepackage{float}

\begin{document}

\title{A Gradient Descent Method for Optimization of Model Microvascular Networks}

\author{Shyr-Shea Chang%
\thanks{Dept. of Mathematics, University of California Los Angeles, Los Angeles, CA 90095, USA.}%
\and
Marcus Roper%
\footnotemark[1]
\thanks{Dept. of Biomathematics, University of California Los Angeles, Los Angeles, CA 90095, USA.}
}

%\author[1]{Shyr-Shea Chang}
%\author[1,2]{Marcus Roper}
%\affil[1]{Dept. of Mathematics, University of California Los Angeles, Los Angeles, CA 90095, USA.}
%\affil[2]{Dept. of Biomathematics, University of California Los Angeles, Los Angeles, CA 90095, USA.}

\maketitle

\begin{abstract}
Within animals, oxygen exchange occurs within networks containing potentially billions of microvessels that are distributed throughout the animal's body. Innovative imaging methods now allow for mapping of the architecture and blood flows within real microvascular networks. However, these data streams have so far yielded little new understanding of the physical principles that underlie the organization of microvascular networks, which could allow healthy networks to be quantitatively compared with networks that have been damaged, e.g. due to diabetes. A natural mathematical starting point for understanding network organization is to construct networks that are optimized accordingly to specified functions. Here we present a method for deriving transport networks that optimize general functions involving the fluxes and conductances within the network. In our method Kirchoff's laws are imposed via Lagrange multipliers, creating a large, but sparse system of auxiliary equations. By treating network conductances as adiabatic variables, we derive a gradient descent method in which conductances are iteratively adjusted, and auxiliary variables are solved for by two inversions of $O(N^2)$ sized sparse matrices. In particular our algorithm allows us to validate the hypothesis that microvascular networks are organized to uniformly partition the flow of red blood cells through vessels. The theoretical framework can also be used to consider more general sets of objective functions and constraints within transport networks, including incorporating the non-Newtonian rheology of blood (i.e. the Fahraeus-Lindqvist effect). More generally by forming linear combinations of objective functions, we can explore tradeoffs between different optimization functions, giving more insight into the diversity of biological transport networks seen in nature.
\end{abstract}

\section{Introduction\label{sec:intro}}

The human cardiovascular network contains billions of vessels, ranging in diameters from centimeters to microns, and continuously carries trillions of blood cells. Cardiovascular networks are robust in some respects and fragile in others. They are robust in the sense that although each network is far more complex than even the largest traffic or hydraulic networks built by humans, in healthy organisms microvascular networks show remarkably little of the chronic patterns of traffic congestion that plague human-built networks. At the same time, the microvascular part of the network; made up of fine vessels less than 8 $\textrm{$\mu$m}$ in diameter, is susceptible to accumulated damage from micro-occlusions\cite{albers2002transient} and micro-aneurysms\cite{klein1995retinal}. This cardiovascular damage is a leading cause of aging related health problems. Systemic microvascular damage associated with diabetes mellitus, can lead to erectile dysfunction\cite{fonseca2005endothelial}, limb loss\cite{pecoraro1990pathways}, neuropathy\cite{reichard1993effect} and dementia\cite{biessels2006risk}. Although each of these forms of microvascular damage is diagnosed and treated in a completely different way, they may have a common physical basis. We therefore ask: What physical functions are microvascular networks organized to perform, and what forms of damage interfere with its ability to perform these functions?

Techniques like plasticization have long enabled the largest vessels in the cardiovascular network to be mapped out. More recently Micro-optical Sectioning Tomography (MOST) has been used to map the blood vessels within rodent brains to micron resolution\cite{wu20143d}, and mapping the blood vessels in the human brain is one of the central goals of the BRAIN initiative\cite{insel2013nih}. Meanwhile long working distance two photon microscopes can be used to directly measure blood flows within living rodent brains\cite{chaigneau2003two,dunn2005spatial}. But using this data still requires understanding of the organizing principles for microvascular networks.

A natural mathematical starting place for deriving organizing principles for transport networks is to frame the problem of network design as a problem in optimization. For example, in 1926 Murray first derived relationships between vessel radii and fluxes at different levels of the arterial network, assuming that the network minimizes a total cost made up of the viscous dissipation and a metabolic cost of maintaining the vessels that is proportional to their volume\cite{murray1926physiological,murray1926physiologicalangle}. A particular consequence of this optimization, is that when a `parent' vessel within the network divides into two `daughters', the sum of the cubes of the daughter radii will equal the cube of the parent radius\cite{sherman1981connecting}, and this result has been validated in studies on real animals\cite{taber2001investigating,sherman1981connecting,zamir1992relation}. The notion of cardiovascular networks as optimizing transport has since found many applications, underlying theoretical models for how energy needs scale with organism size\cite{savage2008sizing,west1997general} as well as clinical computational fluid dynamics (CFD) studies in which different candidate surgical graft geometries are ranked by their transport efficiency\cite{de1996use,marsden2007effects,yang2010constrained}. 

Many (but not all, see Zamir\cite{zamir1992relation}) studies of larger vessels (typically extending down to a few mm in diameter) show that they conform to Murray's law, suggesting that on a population level, these vessels are organized to minimize dissipation. However fine vessels account for a large share of the total network dissipation; for example in humans capillary beds and the arterioles that supply them, account for about a half of the total dissipation in the cardiovascular network\footnote{Since the total flux of blood is the same at each level of the vascular network, we can estimate the dissipation at each level from pressure measurements, such as those summarized in Guyton and Hall\cite{hall2015guyton}}. Yet we are aware of no data that shows that principles of dissipation minimization extend to these vessels, which are typically arranged into topologically complex networks\cite{chaigneau2003two,wu20143d} (also see Figure ~\ref{fig:bionetwork_example}). Indeed our own analysis of the zebrafish trunk microvasculature, which is a model system for studying vascularogenesis, showed that uniform partitioning of red blood cells between the many fine vessels perfusing the trunk, is a more likely candidate optimization principle for these networks than minimizing dissipation\cite{chang2015optimal}. In fact we showed that the adaptations used within the zebrafish trunk network to ensure uniform perfusion directly lead to an 11-fold increase in dissipation within the network\cite{chang2015optimal}.

To understand the function of microvascular networks, and indeed to understand biological transport networks generally, which may be optimized for mixing\cite{alim2013random,roper2013nuclear}, resistance to damage\cite{bebber2007biological,katifori2010damage}, or for the ability to accommodate high variations in flow\cite{katifori2010damage,corson2010fluctuations}, it would be highly useful to have a framework for generating networks that optimize a particular target function, while respecting constraints. Before introducing our method for optimizing general functions we first describe previous methods for generating optimal transport networks (the relationship of this paper to these previous works is also presented in Table 1). Early methods for optimization followed Murray's original approach\cite{durand2006architecture}, by optimizing transport within individual vessels, or at junctions in which single vessels bifurcate\cite{durand2006architecture}. Although these methods allow local geometric optimization -- i.e. of the position and angles of branching points within a network -- they can only be used once the topology of the network, that is, the sequence in which vessels branch or fuse, has been defined. Banavar et al.\cite{banavar2000topology} and Bohn and Magnasco\cite{bohn2007structure} developed an iterative scheme that allowed optimization of entire networks linking a given set of sources to a given set of sinks given constraints on the total amount of material available to build the network. This approach made use of the fact that the laws governing flow in a network (Kirchoff's first and second laws, which will be described in more detail below), are automatically satisfied when dissipation is minimized within a network\cite{durand2007structure}. Katifori et al.\cite{katifori2010damage} and Corson\cite{corson2010fluctuations} later developed this theory to study networks that are designed to minimize dissipation given fluctuating set of source and sink strengths, or under variable damage (in which a random set of links within the network is eliminated). All of these works adopt an iterative approach, in which the conductances of network edges are iteratively updated until the dissipation is minimized: Corson\cite{corson2010fluctuations} uses a relaxation method, while Katifori et al.\cite{katifori2010damage} use gradient descent. In both cases, implicit use is made of the fact that the optimal network (i.e. the one that minimizes dissipation) will also obey Kirchoff's laws. 

Recent advances have focused on how structural adaptation (the process by which vessels within the transport network adjust their radii in response to the amount of flow that they carry) can be used to produce results equivalent to searching for a dissipation minimizing configuration by gradient descent\cite{hu2013adaptation,pries1998structural}. These works also highlight that incorporating both growth and structural adaptation in a network can reliably find global dissipation minimizing configurations (as opposed to locating only local minimizers within a rough landscape)\cite{ronellenfitsch2016global}.

By contrast, the problem of minimizing other functions on networks has received relatively little attention. This is likely because, although there is strong evidence that some biological transport networks, such as fungal mycelia and slime mold tubes\cite{roper2013nuclear,alim2013random}, are adapted to maximize the amount of mixing of the fluids, nutrients and organelles that are transported by the network, microvascular networks have generally been thought to conform to the same principles of dissipation minimization as larger vessels. However, our own work on the embryonic zebrafish vasculature shows that the fine vessels in the trunk are organized to all receive red blood cells at identical rate\cite{chang2015optimal}. Red blood cell partitioning is achieved by increasing the resistance of vessels near the head of the fish over vessels near its tail, leading to a large (11 fold) increase in the dissipation within the network. This study therefore suggests that uniformity of flows, rather than minimization of dissipation, underlies the design of the zebrafish trunk microvasculature. However, our ability to determine whether the principle of flow uniformity may rule in other real networks, or to test alternate candidate optimization principles, is limited because, unlike dissipation, there is no existing method for optimizing general functions that can be evaluated over transport networks. The main mathematical challenge that must be overcome to create such an optimization method is to ensure that in addition to minimizing the given function with given constraints, for example on the total material, the optimal network must respect constraints associated with Kirchoff's laws, which are not automatically satisfied at optima if the function of interest is not the energy dissipation within the network.

Here we devise a method for minimizing arbitrary functions on networks. The method is described in Sections \ref{sec:setup} and \ref{sec:algorithm}. It uses gradient descent that can be rigorously shown to locate local minima of a given function, with a heuristic simulated annealing method, that has previously been shown\cite{katifori2010damage} to be capable of finding global minima in rough landscapes. As a consistency check, we initially use this method to generate networks that minimize dissipation for a given amount of material, checking first that it is consistent with previous results on optimal networks (in Section \ref{subsec:sing_source_sing_sink}), and second showing how these results can be modified if the non-Newtonian rheology of real blood is incorporated into models (in Section \ref{sec:FL}). Then, inspired by our demonstration of uniform flow in the zebrafish vascular network\cite{chang2015optimal}, we go on to minimize a function representing the uniformity of flow within transport networks (Section \ref{sec:unif_flow}), enabling us to calculate the {\it optimal zebrafish trunk vasculature} (Section \ref{sec:zebrafish}). Finally (also in Section \ref{sec:zebrafish}) we use our method to solve for hybrid functionals in which a linear combination of uniformity and dissipation are minimized: allowing the relative priority of uniformity and dissipation to be continuously varied, and allowing us to generate diverse optimal networks to compare with experimental observations.

\begin{table}[htbp]
%\begin{tabularx}{\textwidth}{|X|X|X|}
\hspace*{0cm}
\begin{tabulary}{1.4\textwidth}{|p{4 cm}|p{4cm}|p{6 cm}|}
\hline
Target functional & Constraint & Method \\
\hline
$\sum \frac{Q_{kl}^2}{\kappa_{kl}}$ & $\sum \kappa_{kl}^\gamma$  & local topological optimization\cite{durand2006architecture}, global optimization \cite{bohn2007structure}, structural adaptation\cite{hu2013adaptation}, growth and structural adaptation\cite{ronellenfitsch2016global}, Section \ref{subsec:sing_source_sing_sink}\\
\hline
$\sum \frac{Q_{kl}^2}{\kappa_{kl}}$ with damage and flow fluctuations   & $\sum \kappa_{kl}^\gamma$& global optimization\cite{katifori2010damage}, fluctuating source\cite{corson2010fluctuations}\\
\hline
$\sum \frac{Q_{kl}^2}{\kappa_{kl}}$ &
network volume, including Fahraeus-Lindqvist effect
  & Section \ref{sec:FL}\\
\hline
$\sum \frac{1}{2}Q_{kl}^2$ & $\sum \kappa_{kl}^\gamma$ & Section \ref{sec:Qsqanal}\\
\hline
$\sum \frac{Q_{kl}^2}{\kappa_{kl}}$ & $\sum \left(\kappa_{kl}^\gamma+a\frac{Q_{kl}^2}{\kappa_{kl}}\right)$ & Section \ref{sec:Murray_constaint}\\
\hline
$\sum \frac{1}{2}(Q_{kl}-\bar{Q})^2$ & $\sum \kappa_{kl}^\gamma$ & Section \ref{sec:zebrafish}\\
\hline 
$\sum \frac{1}{2}(Q_{kl}-\bar{Q})^2$ & $\sum \left(\kappa_{kl}^\gamma+a\frac{Q_{kl}^2}{\kappa_{kl}}\right)$ & Section \ref{sec:zebrafish}\\
\hline

\end{tabulary}
\caption{New results presented in this paper, shown with previous works.}
\end{table}

\begin{figure}[h]

	\begin{center}
		\includegraphics[width=10 cm]{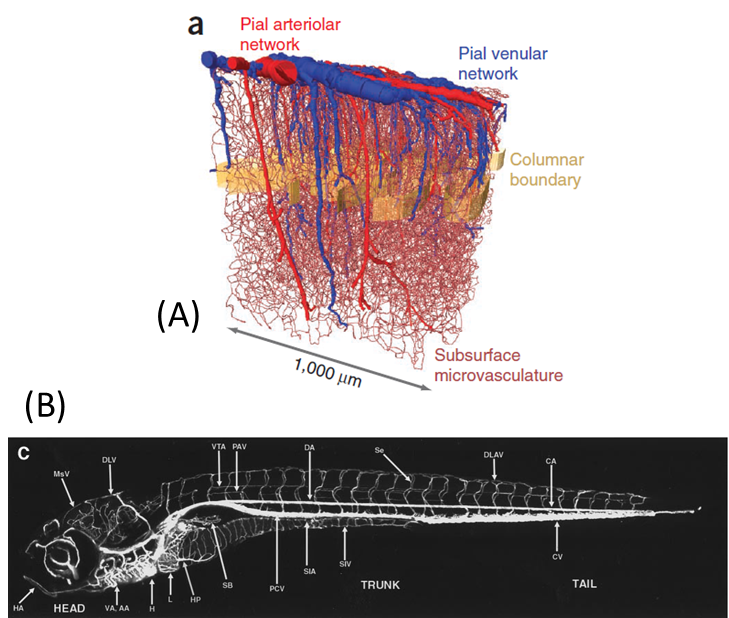}
\caption{Examples of complex microvascular networks. (A) Capillary network in mouse sensory cortex\cite{blinder2013cortical}. (B) Microvascular network of zebrafish 7.5 days post fertilization (dpf) embryo\cite{isogai2001vascular}.}

\label{fig:bionetwork_example}
\end{center}

\end{figure}

\begin{figure}[h]

	\begin{center}
		\includegraphics[width=8 cm]{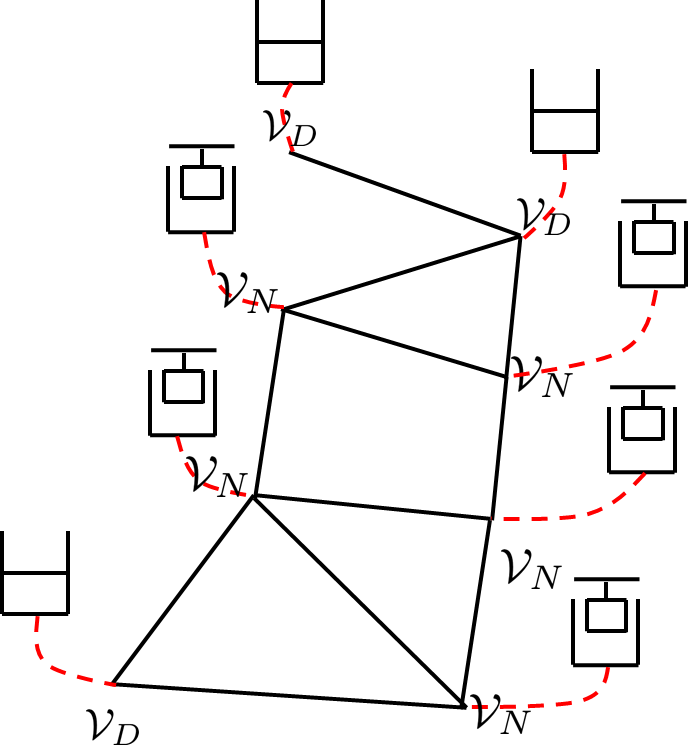}
\caption{Transport network with the Dirichlet vertices $\Vd$ and Newmann vertices $\Vn$. In our representation we imagine Dirichlet vertices connected to fluid reservoirs allowing pressure to be imposed, and Neumann vertices to syringe pumps, allowing inflows or outflows to be imposed}

\label{fig:network_diagram}
\end{center}

\end{figure}

\section{Setup} \label{sec:setup}

First we mathematically frame the problem of finding optimal networks for general network topology. Consider an undirected graph $(\mathcal{V},\mathcal{E})$ with $V$ vertices $k=1,\ldots,V$. For any given 2 nodes $k,l$ we write $\langle k,l\rangle = 1$ if there is a edge linking $k$ and $l$ and $\langle k,l\rangle = 0$ if $k$ and $l$ are not linked. Each edge in the network is assigned a conductance $\kappa_{kl}$; the flow $Q_{kl}$ in the link is then determined by $Q_{kl} = (p_k-p_l)\kappa_{kl}$,
where $p_k$ and $p_l$ are respectively the pressures at the vertices $k$ and $l$. In typical microvascular networks vessel diameters are on the order of 10 $\mu$m, and blood flow velocities are on the order of 1 mm/s, so the Reynolds number, which represents the relative importance of inertia to viscous stresses, is $Re = UL/\nu \approx 4 \times 10^{-3}$, using the viscosity of whole blood $\nu \approx 2.74 \; \textrm{mm}^2/\textrm{s}$. Since $Re \ll 1$ inertial effects may be neglected, and by default the conductances of individual vessels will be obtained from he Hagen-Poiseuille's law\cite{acheson1990elementary}:

\begin{equation}
\kappa = \frac{\pi r^4}{8\mu \ell}
\label{eq:Hagen_Poiseuille}
\end{equation}

\noindent where $\kappa$ is the conductance, $\mu$ is the blood viscosity, $\ell$ is the vessel length, and $r$ is the vessel radius. In ascribing a well-defined pressure to each vertex within the graph, and applying the Hagen-Poiseuille law to compute edge flows from pressures, we assume that there are unidirectional flows within each vessel, ignoring the entrance and exit effects that occur when vessels branch or merge. At moderate Reynolds numbers, entrance effects can strongly affect the flow through vessels, for example by leading to phase separation, whereby red blood cells divide in different ratios at a junction than whole blood\cite{pries2005microvascular}. However, these effects contribute quite weakly for the low Reynolds number flows being modeled in this paper, for example in our previous studies of the zebrafish trunk vascular network, we found that total variation in hematocrit from vessel to vessel was no more than 2-fold. Moreover, we expect the entrance and exit effects to penetrate a distance comparable to the vessel diameter. Since typical microvascular vessels have diameters on the order of 5-10 $\mu$m and lengths on the order of hundreds of $\mu$m, we therefore expect entrance and exit effects to contribute negligibly to the total resistance of the vessel.

The networks we consider consist of vertices and predescribed edges where conductance may be positive (or zero if required by the algorithm) along with two kinds of boundary conditions on vertices (Fig.~\ref{fig:network_diagram}). At any vertice in the network we can either impose Kirchoff's first law (conservation of flux)
\begin{equation}
\sum_{l\,:\,\langle k,l\rangle = 1} Q_{kl} = \sum_{l\,:\,\langle k,l\rangle = 1} \kappa_{kl}(p_k-p_l) = q_k \qquad \forall 1\leq k\leq V~, \label{eq:KirchoffI}
\end{equation}
where $q_k$ is the total flow of blood entering the network (or leaving it if $q_k<0$) at vertex $k$, or we impose $p_k = \bar{p}_k$ (i.e. pressure is specified). We say a node is in $\Vd$ if pressure is specified, or in $\Vn$ if Kirchhoff's first law is imposed, with possible inflow or outflow. This system of $V$ linear equations forms a discretized Poisson equation with Neumann and Dirichlet boundary conditions imposed on selected nodes, and the flow is uniquely solvable if and only if each connected component of the network (connected by edges with positive conductances) either has at least one Dirichlet vertex or $\sum_{k\in \Vn} q_k = 0$ with sum restricted to the component\cite{LP:book}. The general problem that this paper will address is how to tune the conductances within the network to minimize a predetermined objective functional $f(\{p_k\},\{\kappa_{kl}\})$, where $\{p_k\}$ means the set of all $p_k$'s and $\{\kappa_{kl}\}$ denotes the set of all $\kappa_{kl}$'s. Previous works (see Table 1) have shown how to generate networks that minimize the total viscous dissipation occurring within the network: $
f(\{p_k\},\{\kappa_{kl}\}) =\sum_{k>l, \langle k,l\rangle = 1} \kappa_{kl}(p_k-p_l)^2$.

However, the pressures $\{p_k\}$ and conductances $\{\kappa_{kl}\}$ are coupled through Equations (\ref{eq:KirchoffI}). Since the relationship between $\{p_k\}$ and $\{\kappa_{kl}\}$ is holonomic, we may incorporate it into a functional via Lagrange multipliers. The functional that we want to minimize in this paper will take the form:
\begin{align}
\Theta  = & f(\{p_k\},\{\kappa_{kl}\}) + \lambda\left[ \sum_{k>l,\langle k,l\rangle = 1}\left( a\kappa_{kl}(p_k-p_l)^2+ \kappa_{kl}^\gamma d^{1+\gamma}_{kl}\right) - K \right] \nonumber\\
  & - \sum_k \mu_k\left(\sum_{l,\langle k,l\rangle=1}  \kappa_{kl}(p_k-p_l)-q_k\right). \label{eq:Thetadefinition}
\end{align}
which has $V_N +1$ Lagrange multipliers: a set $\{\mu_k|k\in \Vn\}$ enforcing Kirchoff's first law on Neumann vertices (the set $\Vn$ with $|\Vn| = V_N$), and a single multiplier $\lambda$ that constrains the amount of energy that the organism can invest in pushing blood through the network and in maintaining the vessels that make up the network. The transport constraint is made up of two terms: $\sum \kappa_{kl}(p_k-p_l)^2$ represents the total viscous dissipation within the network, while $\sum \kappa_{kl}^\gamma d_{kl}^{\gamma+1}$ represents the total cost of maintaining the network (the {\it material constraint}), with $d_{kl}$ being the vessel length. The exponent $\gamma$ can be altered to embody different models for the cost of maintaining a network. In our default model (Equation \ref{eq:Hagen_Poiseuille}) conductance of an edge is proportional to the fourth power of its radius, so if the cost of maintaining a particular vessel is proportional to its surface area (and thus to its radius), then we expect $\gamma=1/4$, while if the cost is proportional to volume then $\gamma=1/2$. In general we need $\gamma\leq 1$ to produce well posed optimization problems (otherwise, the cost of building a vessel can be indefinitely reduced by subdividing the vessel into finer parallel vessels). Although in (\ref{eq:Thetadefinition}) we initially adopt the same material cost function definition as was used in previous work\cite{katifori2010damage,bohn2007structure}, we will go on to modify the cost function to incorporate networks in which vessels have different lengths, or in which the non-Newtonian rheology of real blood is modeled. Throughout, we incorporate a parameter $a>0$ that represents the relative importance of network maintenance and dissipation to the cost of maintaining the network. When presenting optimal networks, we will discuss the effect of varying $a$ (as well as asymptotic limits in which $a\to 0$) upon the network geometry. Since Murray's work on dissipation-minimizing networks\cite{murray1926physiological,murray1926physiologicalangle} is equivalent to minimizing this constraint function, we will adopt the shorthand of calling the network cost term the {\it Murray constraint}.

Table 1 gives a systematic description of previous work on minimizing functionals across networks, as well as outlining the new results that will be presented here on the optimization of (\ref{eq:Thetadefinition}).

\section{Optimization of general functions on a network by gradient descent\label{sec:algorithm}}

At any local minimum of $\Theta$, each of the partial derivatives of (\ref{eq:Thetadefinition}) must vanish. In order to locate such points, we adopt a gradient descent approach, in which $\kappa_{kl}$ are treated as adiabatically changing variables. That is: $\frac{\partial \Theta}{\partial \kappa_{kl}}$ is calculated, and an optimal perturbation of the form $\delta \kappa_{kl} = -\alpha \frac{\partial \Theta}{\partial \kappa_{kl}}$ is applied to ensure $\Theta$ decreases each time the conductances in the network are updated. At the same time, the other variables in the system, namely $\{p_k,\mu_k,\lambda\}$, are assumed to vary much more rapidly, to remain at a local equilibrium, so that:
\begin{equation}
\frac{\partial\Theta}{\partial p_k} =\frac{\partial\Theta}{\partial \mu_k} =\frac{\partial\Theta}{\partial \lambda}=0~ \label{eq:adiabatic}.
\end{equation}
Our ability to perform gradient descent therefore hinges on our ability to solve the system of $2V_N +1$ equations (\ref{eq:adiabatic}) for each set of conductances $\{\kappa_{kl}\}$ that the network passes through on its way to the local minimum. Fortunately it turns out that only one nonlinear equation in a single unknown variable needs to be solved for to solve all of the conditions (\ref{eq:adiabatic}); the other equations are linear and can be solved with relatively low computational cost.

Because we will consider multiple variants of the Murray constraint, in what follows we will write the summand that enforces the Murray constraint in the general form: $\lambda g(\{p_k\},\{\kappa_{kl}\})$. Then the condition that $\frac{\partial \Theta}{\partial \mu_k}=0$, $k\in \Vn$, merely enforces the system of mass conservation statements at each Neumann-vertex in the network (\ref{eq:KirchoffI}). These equations represent a discretized form of the Poisson equation and can be solved by inverting a sparse $V_N \times V_N$ matrix with $O(E,V_N)$ entries\cite{LP:book}. That is, we write:
\begin{equation}
Dp = f
\end{equation}
where $f_k = q_k$ is the prescribed inflow at Neumann vertices and $f_k = \tilde{p}_k$, the prescribed pressure at Dirichlet vertices. $-D$ is a form of graph Laplacian:
\begin{equation}
D_{kl} \doteq \left\{\begin{array}{llll}
\sum_{l,\langle k,l\rangle =1} \kappa_{kl} & k=l, k\notin \Vd\\
-\kappa_{kl} & \langle k,l\rangle =1, k\notin \Vd\\
\kappa\ub{1} & k=l, k\in \Vd\\
0 & \textrm{otherwise}
\end{array}\right.
\label{eq:graph_Laplacian}
\end{equation}
where $\kappa\ub{1} = 1$. (For any $\kappa\ub{1}\neq 0$ $D$ is full rank; we will make use of other positive constant values for $\kappa\ub{1}$ later.)

To solve for $\{\mu_k\}$, we consider the system of equations $\frac{\partial\Theta}{\partial p_k} = 0$, $k\in \Vn$:
\begin{equation}
0 = \left(\frac{\partial f}{\partial p_k} +\lambda \frac{\partial g}{\partial p_k}\right)- \sum_{l, \langle k,l\rangle =1} (\mu_k-\mu_l)\kappa_{kl}. \label{eq:muequation}
\end{equation}
If $\lambda$, $\{p_k\}$ and $\{\kappa_{kl}\}$ are all known then these equations again take the form of a discrete Poisson equation, however, just as with the solution of the pressure equation, these equations themselves do not admit unique solutions unless a reference value of $\mu_k$ is established. If $\Vd \neq \phi$, i.e. if pressure is specified at least one vertex within $(\V,E)$ then $\mu_k = 0 \; \forall k\in \Vd$ and the $\mu_k$ equations admit a unique solution; otherwise $\mu_k$'s are determined up to a constant (see \ref{App:mu_solve}). For some forms of target function $f$ and constraint function $g$, we will show that $\mu_k$'s for the minimizer are directly related to the pressures, with no need to solve the Poisson equation by a separate matrix inversion.

However, to use Equation (\ref{eq:muequation}) to solve for $\mu_k$ it is still necessary to know the Lagrange multiplier that enforces the Murray constraint (i.e. $\lambda$). The simplest way to derive $\lambda$ is to dictate that the variational of the constraint function should vanish when $\kappa_{kl}$ is updated since the constraint function should remain constant when its variational under changes in conductances, i.e.:
\begin{equation}
0 =  \sum_{k\notin \Vd} \frac{\partial g}{\partial p_k}\delta p_k +\sum_{k>l, \langle k,l\rangle =1} \frac{\partial g}{\partial \kappa_{kl}}\delta \kappa_{kl}
\label{eq:constitutive_variation}
\end{equation}

\noindent (we set $\delta p_k = 0$ if $k\in \Vd$) where

\begin{equation}
\delta\kappa_{kl} = -\alpha \frac{\partial \Theta}{\partial \kappa_{kl}} = -\alpha\left( \frac{\partial f}{\partial\kappa_{kl}}+\lambda \frac{\partial g}{\partial\kappa_{kl}} - \kappa_{kl} (\mu_k-\mu_l)(p_k-p_l)\right).
\label{eq:delkappa}
\end{equation}
At this point $\{\delta p_k\}$ and $\{\mu_k\}$ are undetermined. The lagrange multipliers $\{\mu_k\}$ can be solved in terms of the still unknown $\lambda$ from (\ref{eq:muequation}) (see \ref{App:mu_solve}). The $\{\mu_k\}$ are linear functions of $\lambda$ since (\ref{eq:muequation}) is a linear system. To obtain $\delta p_k$ for each $k\in \Vn$ we calculate the variational in Kirchhoff's first law:

\begin{equation}
\sum_{l, \langle l,k\rangle = 1} \delta\kappa_{kl}(p_k-p_l) + \kappa_{kl}(\delta p_k - \delta p_l) =0. \label{eq:Kirchfirstvar}
\end{equation}

\noindent When written in matrix form, the matrix multiplying $\{\delta p_k\}$ is again the negative of the graph Laplacian, $-D$. Thus $\{\delta p_k\}$ can be solved in terms of $\lambda$ so long as the original matrix system is solvable for $\{p_k\}$. Since $\{\mu_k\}$ are linear in $\lambda$, $\{\delta p_k\}$ are also linear in $\lambda$, which implies that the right hand side of Equation (\ref{eq:constitutive_variation}) is linear in $\lambda$. Therefore $\lambda$ can be solved in closed form from Equation (\ref{eq:constitutive_variation}), and the optimal variation $\delta\kappa_{kl}$ can be determined from equation (\ref{eq:delkappa}).

With $\{p_k\}$, $\{\mu_k\}$, and $\lambda$ solvable given $\{\kappa_{kl}\}$ we can perform gradient descent using Equation (\ref{eq:delkappa}) and numerically approach a minimizer. However our descent method has the following limitations: 1. For finite step sizes $\alpha$, conductances may drop below 0 when perturbed according to Equation (\ref{eq:delkappa}). 2. The method only conserves the Murray function up to terms of $O(\delta\kappa)$. 

To avoid negative conductances we truncate at a small positive value $\epsilon$ at each step, i.e. set:

\begin{equation}
\kappa\ub{n+\frac{1}{2}}_{kl} = \max\{\kappa\ub{n}_{kl} - \alpha \frac{\partial\Theta}{\partial\kappa_{kl}},\epsilon\}.
\end{equation} 

To ensure that the constraint is exactly obeyed we then project the conductances $\{\kappa_{kl}\ub{n+\frac{1}{2}}\}$ onto the constraint manifold $g(\{p_k\},\{\kappa_{kl}\}) = 0$, via a projection function:

\begin{equation}
\kappa\ub{n+1}_{kl} = h(\kappa\ub{n+\frac{1}{2}}_{kl}) \quad \forall \langle k,l\rangle =1, k>l.
\end{equation}

\noindent Throughout this work we consider three possible projection functions: One choice is to project according to the normal of the constraint surface:

\begin{equation}
\kappa\ub{n+1}_{kl} = \kappa\ub{n+\frac{1}{2}}_{kl} - \beta\frac{\partial g}{\partial \kappa_{kl}}(\{p_k\ub{n+\frac{1}{2}}\},\{\kappa_{kl}\ub{n+\frac{1}{2}}\}),\quad \forall \langle k,l\rangle =1,k>l
\end{equation}

\noindent The value of $\beta$ must be chosen numerically to ensure that $g(\{p\ub{n+1}_k\}, \{\kappa_{kl}\ub{n+1}\}) = 0$ exactly. This entails recomputing the pressure distribution $\{p\ub{n+1}_k\}$ for each $\beta$ value, and secant search on $\beta$ to obtain the root. Another approach we have followed is varying the parameter $\lambda$. This method has comparable complexity to projection on $\{\kappa\ub{n+\frac{1}{2}}_{kl}\}$; since the $\{\mu_k\}$ depend linearly on $\lambda$ via Equation (\ref{eq:muequation}), $\{\kappa\ub{n+1}_{kl}\}$ depends linearly on the parameter $\lambda$. However, just as with the projection method, we must still recompute the $\{p\ub{n+1}_k\}$ for each trial set of $\{\kappa\ub{n+1}_{kl}\}$. Moreover, for some target functions $f$ or constraint functions $g$, it is difficult to derive closed-form expressions for $\lambda$ (i.e. to calculate the partial derivatives $\frac{\partial f}{\partial p_k}$ and $\frac{\partial g}{\partial p_k}$). In this case $\lambda$ may only be computed numerically, by solving $g(\{p_k\ub{n+1}(\lambda)\},\{\kappa_{kl}\ub{n+1}(\lambda)\}) = 0$. A third approach that we have adopted is to simply scale the conductances:

\begin{equation}
\kappa\ub{n+1}_{kl} = \beta\kappa\ub{n+\frac{1}{2}}_{kl},\quad \forall \langle k,l\rangle =1,k>l
\label{eq:scaleproj}
\end{equation}

\noindent where $\beta$ is chosen to satisfy the Murray constraint. This method produces theoretically suboptimal corrections on the conductances, but it is typically easy to compute a value of $\beta$ that satisfies the Murray constraint. In particular, under certain boundary conditions, e.g. $p_k = \bar{p}, \; \forall k\in \Vd$ within each connected component of the network meaning that all pressure vertices within a single connected component have the same imposed pressures, a rescaling of the conductances throughout the network leaves the fluxes on each edge unaffected. In this case, the dissipation decreases in inverse proportion to $\beta$, while the maintenance cost increases proportionately to $\beta^\gamma$.

\section{Minimizing dissipation \label{sec:sqgrid}}

\subsection{Single source, single sink networks \label{subsec:sing_source_sing_sink}}

As a first test for our optimization method we recompute dissipation minimizing networks; that is we set $a=0$, so our constraint function only reflects the total material cost of the network, and set the target function equal to $\sum_{k>l,\langle k,l\rangle =1} \kappa_{kl}(p_k-p_l)^2$ so that our algorithm finds the minimal dissipation among all networks built using a given quantity of material. Our base network is a square grid (Fig.~\ref{fig:diss_gamma_1s2_square_grid}A). In addition to allowing for simple vertex indexing, this architecture resembles the regular capillary bed networks observed, for example in the rat gut\cite{tuma2003transcytosis}. We impose an inflow boundary condition on the upper left corner and a fixed zero pressure on the lower right corner. The dissipation-minimizing network is a single geodesic (i.e. path) between source and sink, allowing us to benchmark our optimization method's ability to find known global optima. To test our gradient descent method we form the function:

\begin{equation}
\Theta = \sum_{\langle k,l\rangle =1, k>l}  \kappa_{kl} (p_k-p_l)^2 + \lambda(\sum_{\langle k,l\rangle =1, k>l} \kappa_{kl}^\gamma - K^\gamma) - \sum_{k\notin \Vd} \mu_k \Big(\sum_{l,\langle k,l\rangle =1} \kappa_{kl}(p_k-p_l) - q_k\Big).
\label{func:diss}
\end{equation}
Here we ignore $d_{kl}$ since we assume all the vessels have the same length which may be scaled to 1 by choice of units. The adiabatic variation of $p_k$ and $\mu_k$ is derived from

\begin{equation}
\frac{\partial\Theta}{\partial p_k} = \sum_{l,\langle k,l\rangle =1} 2\kappa_{kl}(p_k-p_l) - \sum_{l,\langle k,l\rangle =1} \kappa_{kl}(\mu_k-\mu_l),\qquad k\notin \Vd
\label{eq:diss_partialpressure}
\end{equation}

\noindent and the fixed pressure boundary condition on pressure nodes allows us to specify that:

\begin{equation}
\mu_i = 0\qquad \forall i\in \Vd
\end{equation}

\noindent The $\mu_k$ are therefore solving a variant of the Kirchhoff flux conservation equations:

\begin{equation}
D \mu = 2 D p
\end{equation}
with $D$ as defined in Equation \ref{eq:graph_Laplacian}.

This system can be solved for $\mu_k$ under the same conditions as the presure equations being solvable (see \ref{App:mu_solve}). In particular if, as here, the only pressure boundary conditions imposed at vertices in $\Vd$ are of the form $p=0$, then $\mu_k = 2p_k , \forall k \in \V$, i.e. $\mu_k$'s exactly represent the pressures for a stationary network. Now we calculate the derivatives with respect to the conductances:
\begin{equation}
\frac{\partial\Theta}{\partial\kappa_{kl}} = (p_k-p_l)^2 + \lambda \gamma \kappa_{kl}^{\gamma-1} - (\mu_k-\mu_l)(p_k-p_l) = \lambda\gamma\kappa^{\gamma-1}_{kl} - (p_k-p_l)^2.
\label{eq:diss_partialkappa}
\end{equation}
In general we determine $\lambda$ from Equations (\ref{eq:constitutive_variation},\ref{eq:delkappa},\ref{eq:Kirchfirstvar}). However the constraint function $g$ is independent of $\{p_k\}$ in this case, so Equation (\ref{eq:constitutive_variation}) becomes
\begin{equation}
0 = \sum_{k>l,\langle k,l\rangle =1} \frac{\partial g}{\partial\kappa_{kl}} \delta \kappa_{kl}
\label{eq:constitutive_variation_no_p}
\end{equation}
and we can solve $\lambda$ directly in terms of $\{p_k\},\{\kappa_{kl}\}$:
\begin{equation}
\lambda = \frac{\sum_{\langle k,l\rangle =1,k>l} \kappa^{\gamma-1}_{kl}(p_k-p_l)^2}{\sum_{\langle k,l\rangle =1,k>l} \gamma \kappa^{2\gamma-2}_{kl}}
\label{eq:diss_lam}
\end{equation}

\begin{figure}[h]

	\begin{center}
		\includegraphics[width=10 cm]{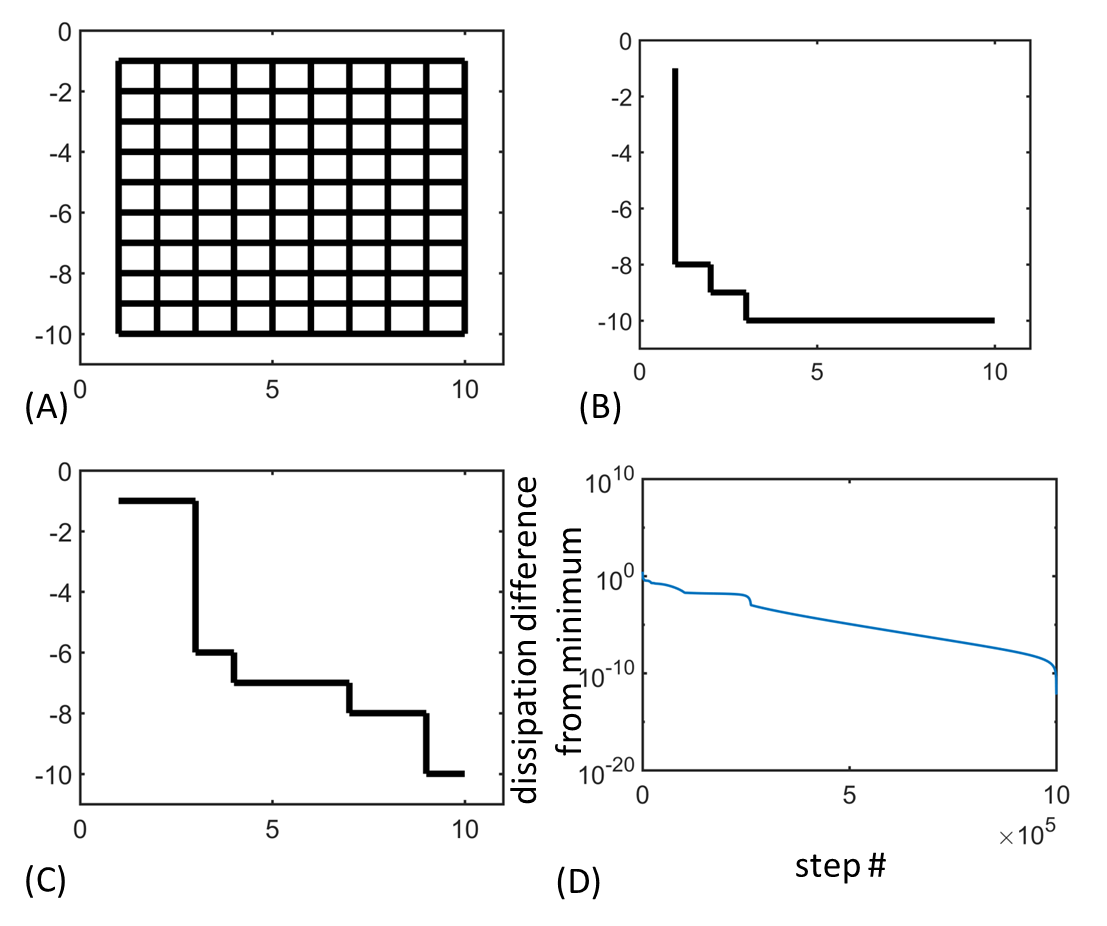}
\caption{Optimal network of the dissipation $\sum \kappa_{kl}(p_k-p_l)^2$ with material constraint $\sum \kappa^\gamma_{kl} = K^\gamma$ and $\gamma = \frac{1}{2}$ on a 10$\times$10 square grid. (A) We represent the capillary bed network by a square grid. (B, C) Different initial conductances produce different optimal networks, but all optimal networks are made of a single wide conduit. Here we use a constant step size throughout the process, and at each step we project by surface normal to maintain the material constraint. (D) The gradient descent algorithm shows a linear convergence, as shown by the dissipation time course of (C).}

\label{fig:diss_gamma_1s2_square_grid}
\end{center}

\end{figure}

\noindent As described in Section \ref{sec:algorithm} we project $\{\kappa_{kl}\}$ along $\frac{\partial g}{\partial\kappa_{kl}} = \gamma \kappa_{kl}^{\gamma-1}$ after each step of the algorithm. At each step of the algorithm, we solve for the pressures $p_k$ from the conductances $\{\kappa_{kl}\}$, then the $\mu_k$, and then descend according to Eqn.~(\ref{eq:diss_partialkappa}). Assuming that $\gamma < 1$, our algorithm deletes edges and concentrates conductance on a single linked path of edges that connects source with sink (Fig.~\ref{fig:diss_gamma_1s2_square_grid}B, C). Any linked path that follows one of the equivalent shortest paths from source to sink will minimize dissipation and accordingly different distributions of random initial conductances converge to different optimal networks. Convergence is linear (Fig.~\ref{fig:diss_gamma_1s2_square_grid}D).

\begin{figure}[h]

	\begin{center}
		\includegraphics[width=10 cm]{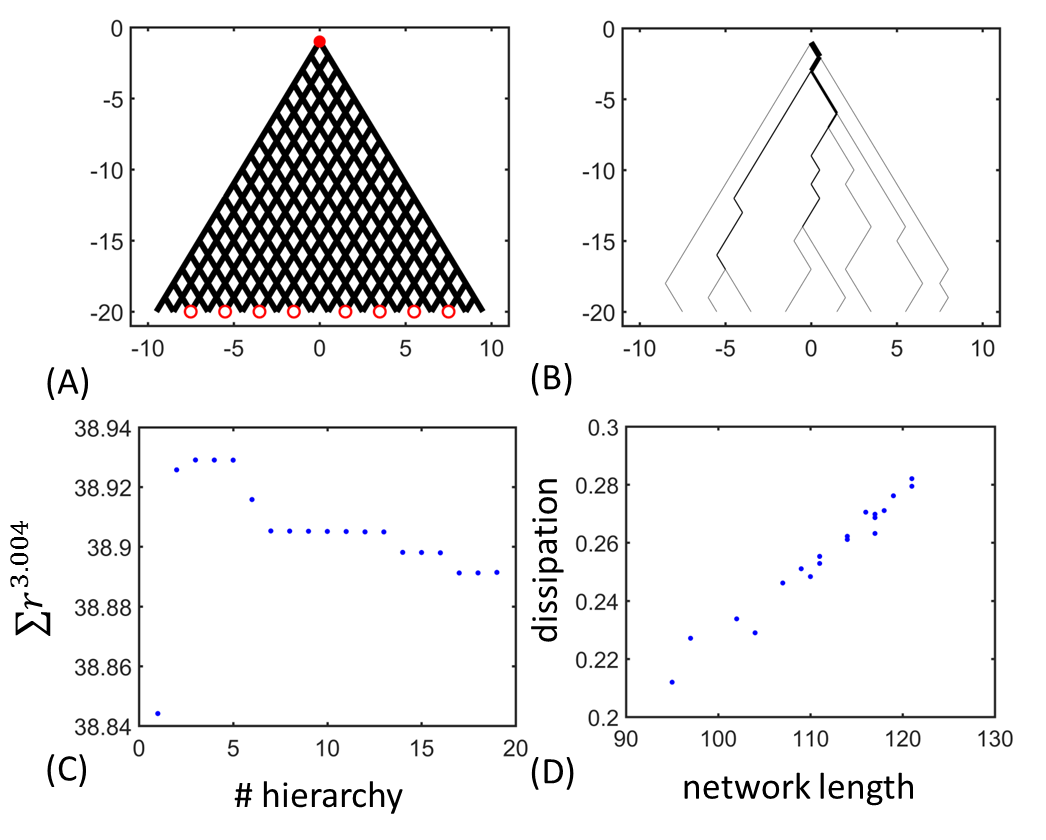}
\caption{Branching network of the dissipation functional $\sum \kappa_{kl}(p_k-p_l)^2$ with material constraint $\sum \kappa^\gamma_{kl}-K^\gamma$ and $\gamma = \frac{1}{2}$. (A) We use a branching grid as our basic topology. There are $N=20$ layers of nodes and a total of $380$ links, connecting a single source (red filled circle) with 8 sinks (red open circles). (B) A minimal dissipative network calculated by gradient descent method. We imposed a fixed zero pressure on the top node and $8$ evenly distributed outflows on the bottom. (C) Murray's law is obeyed by the minimal dissipative network, indicated by the nearly constant sum of radius to an exponent $3.004$ among different hierarchies in network shown in (B). (D) The network length and dissipation vary between different local optima and are strongly correlated with each other (correlation coefficient $r=0.98$).}

\label{fig:diss_gamma_1s2_branch}
\end{center}

\end{figure}

\subsection{Minimizing dissipation with distributed sinks\label{sec:FL}}

The ability of the optimization algorithm to identify shortest distance paths between source and sink is a useful sanity check, but a real test of the algorithm requires that we evaluate its ability to produce known branching tree structures\cite{bohn2007structure,durand2007structure} when the network distributes blood between a single source and multiple, dispersed sinks. We simulate such a network by splitting the grid representing the capillary network in half along the diagonal. The source continues to be one corner of the square, and we space out a number of sinks, with equal output fluxes, along the diagonal (Fig.~\ref{fig:diss_gamma_1s2_branch}A). To make the pressure equation solvable we set pressure at the top-most (source) vertex in the network to $p_1 = 0$. Sink nodes each have prescribed outflows.

\begin{figure}[H]

	\begin{center}
		\includegraphics[width=10 cm]{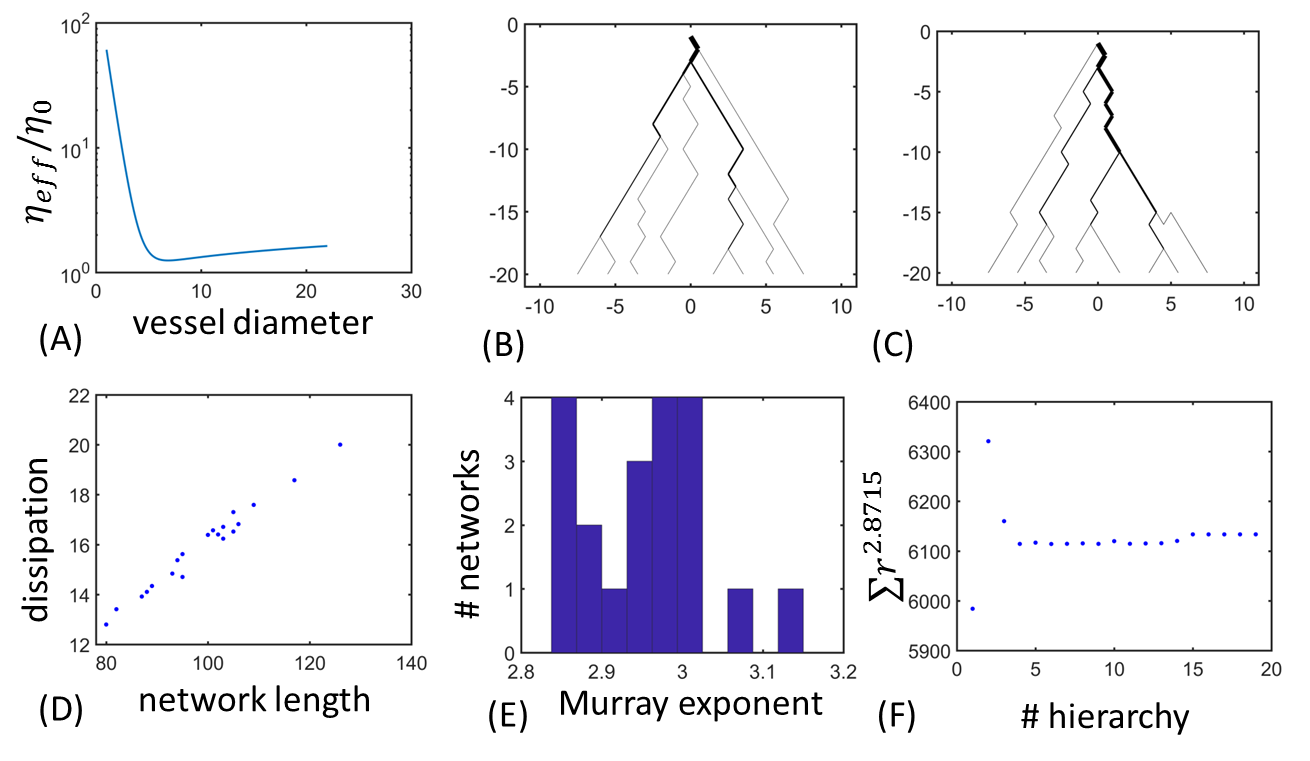}
\caption{Branching network of the dissipation functional with cost functions modified by the Fahraeus-Lindqvist effect. (A) The ratio between effective blood viscosity and plasma viscosity in rodents rises sharply as vessel diameter decreases from around $3\mu m$. (B, C) The optimal networks have tree structure similar to their Newtonian counterparts. We imposed a fixed zero pressure on the top node and $8$ evenly distributed outflows on the bottom. There are $20$ hierarchies and a total of $380$ links. The Fahraeus-Lindqvist effect was modeled according to Pries et al.\cite{pries2005microvascular} and we adjust the multiplicative constant so that a vessel of $6 \; \mu m$ in diameter corresponds to unit conductance, which is the mean of randomized initial conductances. (D) Again the gradient descent method finds local minima and the dissipation is strongly correlated with the network length (correlation coefficient $r = 0.99$). (E) The Murray exponents are in general lower than $3$, the value derived by Murray, and have a wider distribution with mean and standard deviation of $2.95\pm 0.081$. We calculate the exponent based on minimizing the coefficient of variation (CV) in terms of the exponent, and exponents for $20$ optimal networks starting from random initial conductances are plotted. The material constraints are fixed to the same value $\sum_{k>l,\langle k,l\rangle =1} D_{kl}^2 = 1.39\times 10^4$. (F) The power law fits reasonably well to the optimal networks with the Fahraeus-Lindqvist effect, indicated by the sum of radii in each hierarchy to the best exponent $2.87$ for constant fit used for network in (C).}

\label{fig:diss_FL_branch}
\end{center}

\end{figure}

Initially we assume the Hagen-Poiseuille law holds in each edge, so $\Theta$ takes the form specified in Equation (\ref{func:diss}); and we follow the same method for updating conductances as in \ref{subsec:sing_source_sing_sink}. Optimal networks take the form of hierarchical branching trees (i.e. loopless networks\cite{durand2007structure}) (Fig.~\ref{fig:diss_gamma_1s2_branch}B) in which thicker vessels bifurcate into narrower vessels, and thence into even narrower vessels similar to Bohn et al.\cite{bohn2007structure}. We can quantitatively test for the ability of our algorithm to produce locally optimal networks by checking that the networks that it converges to obey Murray's law\cite{murray1926physiological,sherman1981connecting} which states that the flow in each vessel in an optimal dissipation network is proportional to the cube power of the radius of the vessel. Since the total flows through each level $y =$ constant must be equal, Murray's law implies that the sum of the cube of the radii of vessels passing through each level should be equal. To test for local optimality, we calculate a Murray exponent by finding the exponent $a$ that minimizes the variance on $\sum r^a_i$ where sums are taken over each edge in the same level of the network (Fig.~\ref{fig:diss_gamma_1s2_branch}C). The Murray exponents are tightly clustered around 3 ($3.01\pm 0.03$), which agree will with the theoretical value. 

Although our algorithm always converged to a locally optimal transport network, different initial configurations ultimately converged to different optima, consistent with previous results showing that the dissipation function landscape is rough with many local optima. To map out this landscape we measure the total length of the network for different local optima. Total length can be a measure of whether the branch points are concentrated near the source (i.e. small $|y|$, producing longer networks) or near the sinks (i.e. large negative $y$, producing shorter networks). The total length has a large variation among optimal networks and also correlates strongly with the dissipation (Fig.~\ref{fig:diss_gamma_1s2_branch}D, $r = 0.98$). This suggests that while a network with larger total length could be a local minimum, the dissipation can be decreased by a topological change that decreases the number of links, though this requires moving away from the local minimum. This suggests that the roughness of the dissipation landscape is high and a strategy of global optimization such as combining gradient descent with simulated annealing must be implemented to find the global minimal dissipation network (see Katifori et al.\cite{katifori2010damage} and below).

Although the assumption that each blood vessel obeys the Hagen-Poiseuille law is a useful default model, the non-Newtonian nature of blood means that in vessels of different diameters, blood may have very different apparent viscosity. In particular the finest vessels in a cardiovascular network are typically comparable in size to the red blood cells they transport. Red blood cells therefore occlude fine vessels, increasing the effective resistance of these vessels. At the same time, in larger vessels, red blood cells tend to self-organize to flow in the center of the vessel, leaving low viscosity layers of plasma adjacent to the vessel walls, reducing resistance to flow in those vessels. It is usual to incorporate these effects into models of vessel conductance by continuing to assume the Hagen-Poiseuille law (Eqn.~(\ref{eq:Hagen_Poiseuille}))

\begin{equation}
\kappa = \frac{\pi D^4}{128\etaeff(D,\phi) \ell},
\end{equation}

\noindent where $D,\ell$ are the diameter and the length of the vessel and the effective viscosity, $\etaeff$, is as a function of vessel diameter and of the concentration (i.e. volume fraction) of blood cells, $\phi$ \cite{pries2005microvascular}. Assuming that red blood cells are partitioned in the same ratio as the ratio of whole blood fluxes at points at which vessels divide, we may assume that the red blood cell concentration is constant through the network. This assumption excludes the effect of the Zweifach-Fung effect, in which the finite size of red blood cells reduces their probability of entering fine vessels, so that larger vessels tend to also contain higher concentrations of red blood cells\cite{pries1989red,pries2005microvascular}. However our own studies of the zebrafish microvasculature show that hematocrit varies only weakly between micro-vessels while conductance variation between similar vessels (such as between different trunk intersegmental vessels) may exceed a factor of 39. Accordingly we incorporate  an empirical model for the dependence of viscosity upon vessel diameter only. Pries and Secomb\cite{pries2005microvascular} measured apparent viscosity of red cell suspensions by analyzing flow of rodent blood through glass capillaries and found that the effective viscosity could be fit empirically by a function:

\begin{equation}
\etaeff(D) = \left[ 220\exp(-1.3D)+3.2 -2.44\exp(-0.06D^{0.645})\right] \eta_0.
\label{eq:FL_0.45}
\end{equation}

\noindent Here the vessel diameter, $D$, is measured in microns, and $\eta_0$ is the plasma viscosity, which is comparable to water $\eta_0 \approx 1 \; cP$). The functional dependence of $\etaeff$ upon vessel diameter, $D$, is shown in Fig.~\ref{fig:diss_FL_branch}A. 

 We expect Equation (\ref{eq:FL_0.45}) to present a good fit only for blood suspensions where the cell radius and hematocrit are comparable to the experiments of Pries and Secomb. It does not apply therefore to the zebrafish network which we study in Section \ref{sec:zebrafish}. However our algorithm is flexible enough to be able to include different functions in place of Equation (\ref{eq:FL_0.45}): We expect qualitatively similar conclusions to hold for different models for the Fahraeus-Lindqvist effect. Incorporating the Fahraeus-Lindqvist effect requires that we rewrite the material constraint since we can no longer simply obtain the radius, and thus volume, of a vessel from its length and conductance. Instead we write:

\begin{equation}
g(\{\kappa_{kl}\}) =  \sum_{k>l, \langle k,l\rangle =1} d_{kl} D(\kappa_{kl},d_{kl})^2 - K
\label{eq:FL_constraint}
\end{equation}

\noindent where $D(\kappa,d)$ maps from the conductance and length of a vessel to its diameter $D$ (we neglect the factor $\frac{\pi}{4}$ since we can absorb it into $K$). We continute to assume that the vessels all have the same length $\ell = 1$ so we can write $D\equiv D(\kappa)$. Numerically we find that $\kappa(D)$ is an increasing function so the inverse function $D(\kappa)$ exists. The change in cost function does not affect $\frac{\partial\Theta}{\partial p_k}$, so $\mu_k = 2p_k \; \forall 1\leq k\leq V$ still holds. However the conductance derivatives now change to:

\begin{equation}
\frac{\partial \Theta}{\partial \kappa_{kl}} = 2\lambda D_{kl} D'(\kappa_{kl}) - (p_k - p_l)^2 = 2\lambda \frac{D_{kl}}{\kappa'(D_{kl})} - (p_k-p_l)^2
\label{eq:FL_diss_partialkappa}
\end{equation}

\noindent where $D_{kl}$ are the diameters corresponding to $\kappa_{kl}$ according to Equation (\ref{eq:FL_0.45}). $\lambda$ can be solved solely from Equation (\ref{eq:constitutive_variation_no_p}):

\begin{equation}
\lambda = \frac{\sum_{\langle k,l\rangle=1,k>l} \frac{(p_k-p_l)^2}{\kappa'(D_{kl})} D_{kl}}{\sum_{\langle k,l\rangle=1,k>l} \frac{2}{\kappa'(D_{kl})^2} D^2_{kl}}.
\end{equation}

\noindent The projection works in the same manner:

\begin{equation}
n_{kl} = 2 D_{kl} D'(\kappa_{kl}) = \frac{2 D_{kl}}{\kappa'(D_{kl})}.
\end{equation}

For these networks we found a much larger number of local optima than when flow convergence was assumed to be Newtonian. To deal with these optima and accelerate convergence we adopt one part of the simulated annealing method of Katifori et al.\cite{katifori2010damage}. Specifically, when the change in conductance ($\max\{|\kappa\ub{n+1}-\kappa\ub{n}|\}$) becomes too small (in practice we adopt a thrshold of $10^{-3}$, then we multiply all conductances (above threshold $\epsilon$) in the network by a multiplicative noise. Then among all the local minimum visited we select the network with the smallest dissipation. The morphology of non-Newtonian minimally dissipative networks qualitatively resembles Newtonian ones in the sense that they are trees (Fig.~\ref{fig:diss_FL_branch}B, C). A strong correlation between dissipation and the total length of the network is again observed (Fig.~\ref{fig:diss_FL_branch}D, $r=0.99$). Here the material of an edge is no longer a certain power of conductance, which is the basis for the original derivation of Murray's law\cite{murray1926physiological,sherman1981connecting}. Therefore we expect that the Murray's exponent, defined again by minimization of variance in $\sum_i r^a_i$, might be far from the theoretical value for Newtonian minimially dissipative networks. However we find that the here the Murray's exponents ($2.95 \pm 0.081$) are quite close to 3, the theoretical value for Newtonian networks, and the sum $\sum_i r^a_i$ is well approximated by constant with the optimized exponent $a$ (Fig.~\ref{fig:diss_FL_branch}E, F). It has been proven for Newtonian flow\cite{durand2007structure} under general boundary conditions\cite{chang2017minimal} that optimal networks are simply connected. However this proof hinges on the fact that Newtonian flows within a network minimize dissipation (or a related quantity called the complementary dissipation\cite{chang2017minimal}). This result does not directly translate to the non-Newtonian flows, including the one described by Equation (\ref{eq:FL_0.45}). Our numerical result supports that minimally dissipative networks with the Fahraeus-Lindqvist effect are trees and satisfy Murray's law, but further theoretical work will be needed to
confirm that this model for the Fahraeus-Lindqvist effect always produces simply connected optimal networks, or to show that optimal networks are generally simply connected even when other non-Newtonian features of blood (such as the Zweifach-Fung effect) are incorporated.

\section{Optimizing uniformity of flow} \label{sec:unif_flow}

\subsection{Optimizing uniformity of flow with material constraint\label{sec:Qsqanal}}

Analyzing minimal dissipation on networks allowed us to compare the performance of the algorithm described in this paper with previous work. We now turn to other target functions that have not been extensively studied. At the level of micro-vessels it is likely that oxygen perfusion rather than transport efficiency is the dominant principle underlying network organization. Indeed our own studies of the embryonic zebrafish trunk vasculature\cite{chang2015optimal} showed that red blood cells are uniformly partitioned among different trunk microvessels, and that the "cost" of uniform perfusion (in the sense of the increase in dissipation over a uniform network that did not uniformly perfuse the trunk) was an 11-fold increase in dissipation. We therefore frame this question more generally, i.e. ask what organization of vessels achieves a given amount of flow $\bar{Q}$ on all links or equivalently, how the flow variation

\begin{equation}
f(\{p_k\},\{\kappa_{kl}\}) = \sum_{\langle k,l\rangle = 1, k>l} \frac{1}{2}(Q_{kl} - \bar{Q})^2
\end{equation}

\noindent may be minimized by optimal choice of conductances $\kappa_{kl}$. We can expand the function $f$ and abandon the constant term:

\begin{equation}
f(\{p_k\},\{\kappa_{kl}\}) = \sum_{k>l, \langle k,l\rangle = 1} \Big(\frac{1}{2}Q^2_{kl} - \bar{Q}Q_{kl} \Big).
\end{equation}

\noindent Under the assumption that the total flow on all edges is conserved, i.e.:

\begin{equation}
\sum_{\langle k,l\rangle =1,k>l} Q_{kl} = C 
\label{eq:hierarchy_assumption}
\end{equation}

\noindent the function f can be reduced to 

\begin{equation}
f(\{p_k\},\{\kappa_{kl}\}) = \sum_{\langle k,l\rangle = 1, k>l} \frac{1}{2}Q^2_{kl}
\end{equation}

\noindent by ignoring constants. The assumption (\ref{eq:hierarchy_assumption}) is valid in networks provided that the network may be divided into levels: that is a series of control surfaces may be constructed between source and sink, with no pair of control surfaces intersecting and each edge intersected by one control surface (Fig.~\ref{fig:diagram_control_surface}). Then since the total flow across each control surface is the same, the total flow over all network edges is $\sum_{k>l,\langle k,l\rangle =1} Q_{kl} = SF$ where $F$ is the total sink strength and $S$ is the number of control surfaces. Both symmetric branching trees and quadrilateral grids (such as the one shown in Fig.~\ref{fig:diagram_control_surface}) are examples of networks having this property, and both can be used as simplified models of microvascular transport networks\cite{hu2012blood}. Without any constraint the function to be optimized can now be written as

\begin{equation}
\Theta = \sum_{\langle k,l\rangle =1,k>l} \frac{1}{2}(p_k-p_l)^2 \kappa_{kl}^2 - \sum_{k\in \Vn} \mu_k \left(\sum_{l,\langle k,l\rangle =1} \kappa_{kl}(p_k-p_l)-q_k \right).
\label{functional:unifsimp}
\end{equation}

Here we show that the optimal networks optimizing (\ref{functional:unifsimp}) have the same flow as the network with uniform conductances, although many different sets of conductances lead to the same set of flow.

\begin{figure}[h]

	\begin{center}
		\includegraphics[width=5 cm]{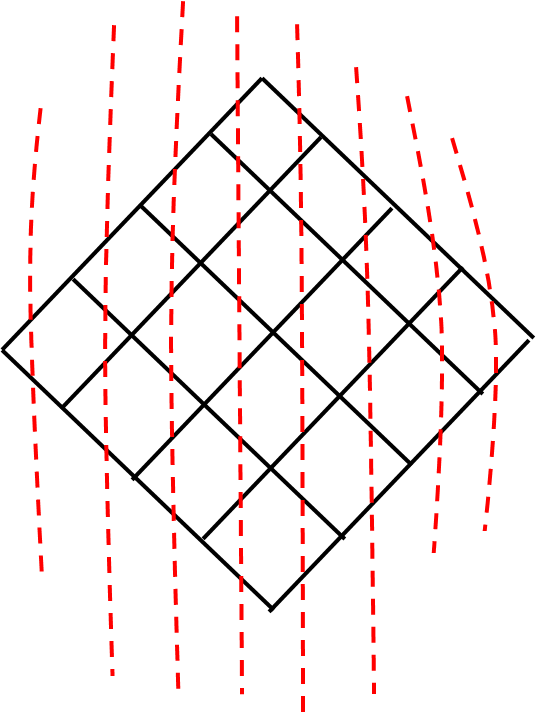}
\caption{A quadrilateral grid (black) can be divided using a set of non-intersecting control surfaces (red dashed lines) such that each edge in the grid is intersected by exactly one control surface.}

\label{fig:diagram_control_surface}
\end{center}

\end{figure}

\begin{theorem}
A stationary network of the functional (\ref{functional:unifsimp}) in which $p_k = 0 \;\forall k\in \Vd$ has the same set of flows as a uniform conductance network with the same support on edges. That is, suppose we let $\kappa_{kl}, Q_{kl}$ be the conductances and flows on the stationary network, and $\kappa'_{kl}, Q'_{kl}$ be those on the uniform conductance network, i.e.

\begin{equation}
\kappa'_{kl} = \left\{\begin{array}{ll}
1 & \mbox{if $\langle k,l\rangle =1$ with $\kappa_{kl}>0$}\\
0 & \mbox{if $\langle k,l\rangle =1$ with $\kappa_{kl}=0$}\\
\end{array}\right..
\end{equation}
Then

\begin{equation}
Q_{kl} = Q'_{kl} \qquad \forall \langle k,l\rangle =1.
\end{equation}

\end{theorem}

\begin{proof}

The assumption that all pressure vertices have pressure zero is really an assumption that all pressure vertices have the same pressure: In the latter case since a constant shift in all pressures does not change the flows. To find the critical points of $\Theta$ we calculate the derivatives:

\begin{equation}
\frac{\partial \Theta}{\partial p_k} = \sum_{l,\langle k,l\rangle=1} \kappa^2_{kl}(p_k-p_l) - \sum_{l,\langle k,l\rangle=1} (\mu_k - \mu_l)\kappa_{kl},\qquad k\notin \Vd
\label{eq:unifsimp_partialp}
\end{equation}

\begin{equation}
\frac{\partial\Theta}{\partial \kappa_{kl}} = \kappa_{kl}(p_k-p_l)^2 - (\mu_k - \mu_l)(p_k-p_l)
\label{eq:unifsimp_partialkappa}
\end{equation}

\noindent along with $\mu_i = 0 \; \forall i\in \Vd$ by assumption. Now we show that a uniform distribution of conductances would result in a critical point $(\{p_k\},\{\mu_k\},\{\kappa_{kl}\})$, by rewriting the equation $\frac{\partial \Theta}{\partial p_k} = 0$ (\ref{eq:unifsimp_partialp}) into the matrix form:

\begin{equation}
D\mu = D\ub{2}p.
\label{eq:unifsimp_unifkappa_matrixeq}
\end{equation}

\noindent Here $D_{kl}$ is in Equation (\ref{eq:graph_Laplacian}) and $-D\ub{2}$ is another graph Laplacian:

\begin{equation}
D_{kl}\ub{2} \doteq \left\{\begin{array}{llll}
\sum_{l,\langle k,l\rangle =1} \kappa_{kl}^2 & k=l, k\notin \Vd\\
-\kappa_{kl}^2 & \langle k,l\rangle =1, k\notin \Vd\\
\kappa\ub{2} & k=l, k\in \Vd\\
0 & otherwise
\end{array}\right.
\label{eq:unifsimp_kappasqmatrix}
\end{equation}
in which the matrix is made full-rank if $\kappa\ub{2}>0$ (similarly to the $\kappa\ub{1}$ constant in $D$). The $\kappa\ub{1}$ entries in $D_{kl}$ enforce $\mu_k =0$ at each $k\in \Vd$. The entries in $D\ub{2}$ are not needed since $p_k=0$ at each $k\in \Vd$, but we add values here to emphasize the symmetry between $\{\mu_k\}$ and $\{p_k\}$. Now consider uniform conductances, i.e. $\kappa_{kl} = a>0 \; \forall \langle k,l\rangle =1$. We can set $\kappa\ub{1} = a$ and $\kappa\ub{2} = a^2$. Then we have $D = aD\ub{2}$ and since $D$ is invertible (see \ref{App:mu_solve})

\begin{equation}
\mu = D^{-1} D\ub{2} p = ap.
\end{equation}

\noindent Now this set of $\mu_k$'s and $p_k$'s then also satisfies $\frac{\partial\Theta}{\partial\kappa_{kl}} = 0$ because

\begin{equation}
\frac{\partial \Theta}{\partial \kappa_{kl}} = a(p_k-p_l)^2 - a(p_k-p_l)^2  = 0.
\end{equation}

\noindent Thus the network with uniform conductances along with pressures solved from the Kirchhoff's first law is indeed a critical point. 

Now we show that any interior critical point, i.e. satisfying $\kappa_{kl}>0 \; \forall \langle k,l\rangle =1$, has the same flows as the uniform conductance network. We will see that for any such network the $\mu_k$'s represent the pressures of the uniform conductance network. Since all the conductances are positive we have $\frac{\partial\Theta}{\partial\kappa_{kl}} = 0 \; \forall \langle k,l\rangle=1$. Assume for now $p_k - p_l \neq 0 \; \forall \langle k,l\rangle =1$. Then from Equation (\ref{eq:unifsimp_partialkappa}) we obtain that the $\{\mu_k\}$ obey a system of equations

\begin{equation}
\kappa_{kl}(p_k-p_l) - (\mu_k-\mu_l) = 0,\qquad \forall \langle k,l\rangle =1
\label{eq:unifsimp_minimizer_condition}
\end{equation}

\noindent which may be rewritten as

\begin{equation}
\mu_k - \mu_l = \kappa_{kl}(p_k - p_l) = Q_{kl},\qquad \forall \langle k,l\rangle = 1.
\end{equation}
Kirchhoff's first law in terms of $\mu_k$'s then reads

\begin{equation}
\sum_{l,\langle k,l\rangle =1} (\mu_k - \mu_l) = q_k \qquad \forall k\in \Vn,\qquad \mu_k = 0 \qquad \forall k \in \Vd.
\end{equation}

\noindent In matrix form the equations can be written as

\begin{equation}
D\mu = F
\label{eq:unifsimp_matrix_eq}
\end{equation}

\noindent where $F_k = q_k$ if $k\in \Vn$ and is zero otherwise, and $D$ is defined as for network made up of unit conductances:

\begin{equation}
D_{kl} \doteq \left\{\begin{array}{llll}
\sum_{l,\langle k,l\rangle =1} 1 & k=l, k\notin \Vd\\
-1 & \langle k,l\rangle =1, k\notin \Vd\\
1 & k=l, k\in \Vd\\
0 & o.w.
\end{array}\right.
\label{eq:unifsimp_matrix}
\end{equation}
Because $D$ is invertible we can solve for $\mu_k$'s from Eqn.~(\ref{eq:unifsimp_matrix_eq},\ref{eq:unifsimp_matrix}). The $\{\mu_k\}$'s represent the pressures that would occur at each vertex if all conductances in the network were set equal to 1, creating uniform conductance network. Since the flows $Q_{kl} = \mu_k - \mu_l$ are determined by $\mu_k$'s we conclude that the locally optimal networks would have flows the same as in the network of uniform conductances. 

To derive (\ref{eq:unifsimp_minimizer_condition}) from (\ref{eq:unifsimp_partialkappa}) we had to assume that $p_k\neq \_l$ whenever $\langle k,l\rangle=1$. Consider the case where in the optimal network $p_k - p_l = 0$ for some $\langle k,l \rangle =1$. For these $(k,l)$'s Eqn.~(\ref{eq:unifsimp_minimizer_condition}) no longer holds and we have to set $\frac{\partial\Theta}{\partial p_k} = 0$ in Eqn.~(\ref{eq:unifsimp_partialp}) to obtain extra information. We claim that $\mu_k = \mu_l$ if $p_k - p_l = 0$. This can be seen from a loop current argument similar to that used in \ref{App:mu_solve} to prove existence and uniqueness of the $\{\mu_k\}$. Specifically, suppose for contradiction that $\mu_{k_1}\neq \mu_{k_2}$ for some pair of vertices with $p_{k_1}-p_{k_2} = 0$ and without loss of generosity let $\mu_{k_1} > \mu_{k_2}$. If $k_1$ and $k_2 \in \Vd$ then $\mu_{k_1} = \mu_{k_2} = 0$; so at least one of the two vertices does not lie in $\Vd$. If $k_2\notin \Vd$ then $\frac{\partial\Theta}{\partial p_{k_2}} = 0$ implies:

\begin{equation}
\sum_{l,\langle k_2,l\rangle =1} \kappa^2_{k_2 l}(p_{k_2}-p_l) = \sum_{l,\langle k_2,l\rangle =1} \kappa_{k_2 l}(\mu_{k_2}-\mu_l).
\end{equation}

\noindent Since Eqn.~(\ref{eq:unifsimp_minimizer_condition}) holds when $p_k-p_l\neq 0$ we have

\begin{equation}
0 = \sum_{l,\langle k_2,l\rangle =1, p_l = p_{k_2}} \kappa^2_{k_2 l}(p_{k_2} -p_l) = \sum_{l,\langle k_2,l\rangle =1, p_{k_2}= p_l} \kappa_{k_2 l}(\mu_{k_2}-\mu_l).
\end{equation}

\noindent Since $\kappa_{kl}>0 \; \forall \langle k,l\rangle =1$ and the sum includes the negative summand $\kappa_{k_2 k_1}(\mu_{k_2}-\mu_{k_1})$ we can find $l$ for which $\mu_l < \mu_{k_2}$ and $p_l = p_{k_2}$. We let $k_3 = l$ and repeat the process to find a neighbor of $k_3$ such that $p_l = p_{k_3}$ but $\mu_l < \mu_{k_3}$. We then can keep repeating this process until we reach a vertex $k_N \in \Vd$ (no vertex may be visited more than once). We have imposed $\mu_{k_N}=0$. Now we trace through increasing $\mu_k$'s starting from $k_2$ and $k_1$ and we get $k'_1,...,k'_{N'}$ such that $\mu_{k'_n} < \mu_{k'_{n+1}} \; \forall n=1,...,N'-1$ and $\mu_{k'_1}>\mu_{k_1}$. By the same reasoning we have $k_{N'} \in \Vd$ and we reach a contradiction since $0 = \mu_{k'_{N'}} > \mu_{k'_{N'-1}} >\cdots > \mu_{k'_1} > \mu_{k_1}>\cdots>\mu_{k_N} = 0$. Therefore $\mu_k = \mu_l$ when $p_k = p_l$ and Eqn.~(\ref{eq:unifsimp_minimizer_condition}) actually holds for all $\langle k,l\rangle = 1$. Again we conclude that the flows of a locally optimal network with non-zero conductances are the same as the flows in the uniform conductance network. 

Finally we discuss the boundary case where $\kappa_{kl} = 0$ for some $\langle k,l\rangle = 1$, and we denote this set of links by $I$. To avoid ill-posedness of pressures we require that that the matrix $D$ is invertible. In this case we do not have Eqn.~(\ref{eq:unifsimp_minimizer_condition}) for $\kappa_{kl} = 0$ because $\frac{\partial\Theta}{\partial\kappa_{kl}}$ need not be zero on these edges. However since there is no flow through links with $\kappa_{kl} = 0$ we can write down Kirchhoff's first law as

\begin{equation}
D\mu = 0,
\label{eq:unifsimp_boundarynetwork_matrixform}
\end{equation}
where $-D$ is again the graph Laplacian, but with zero conductance edges removed and other edges with conductance $1$:

\begin{equation}
D_{kl} = \left\{\begin{array}{llll}
\sum_{l,\langle k,l\rangle =1, (k,l)\notin I} 1 & k=l,k\notin \Vd\\
-1 & \langle k,l\rangle =1, (k,l)\notin I\\
1 & k=l, k\in \Vd\\
0 & \textrm{otherwise}\\
\end{array}\right..
\end{equation}

\noindent We can safely remove the zero conductance links from the network because the difference $\mu_k - \mu_l$ no longer represents the flow $Q_{kl}$, and that we know $Q_{kl} = 0$ for these links. By assumption we can solve for $\mu$ from Eqn.~(\ref{eq:unifsimp_boundarynetwork_matrixform}) so $\{\mu_k\}$ represent the pressures within the uniform conductance network, but with links $\kappa_{kl} = 0$ removed from the network.
\end{proof}

\begin{figure}[h]

	\begin{center}
		\includegraphics[width=15 cm]{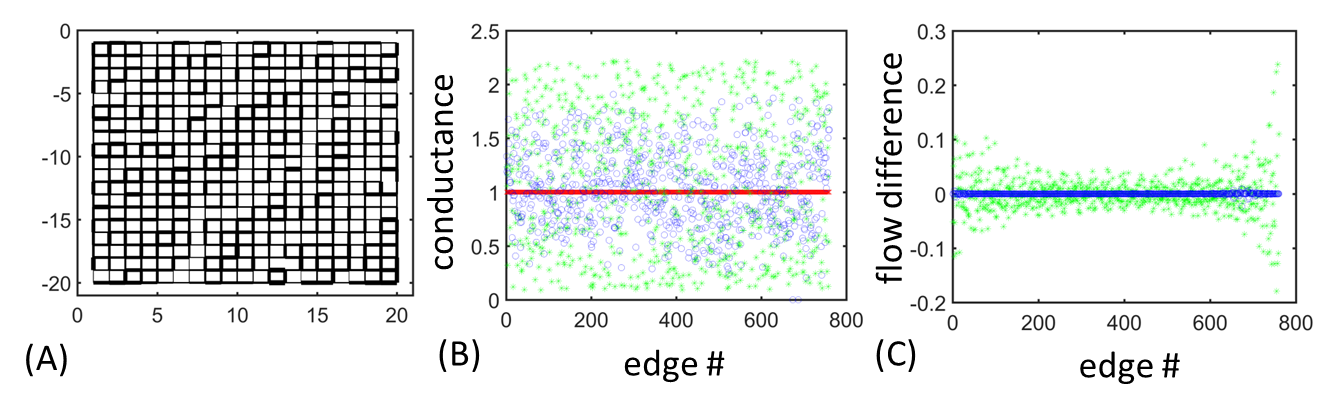}
\caption{Optimal network of functional $\frac{1}{2}\sum Q^2$ on a 20$\times$20 square grid with $400$ nodes. (A) An optimal network has a seemingly random distribution of conductances. The link widths are proportional to the conductances. (B) A closer view reveals that the conductances of the optimal network (blue circle) are quite different from uniform (red cross), and do not seem qualitatively different from initial conductances drawn from a uniform random distribution (green star). The conductances are normalized such that $\sum \kappa^\frac{1}{2}$ are the same. (C) The differences of flows from those in a uniform conductance network (blue circles) are uniformly zero, while the differences from the random initial conductance distribution (green stars) are not.}

\label{fig:Qsq_square_grid}
\end{center}

\end{figure}

Finally we numerically calculate the optimal network for uniformizing flow to verify the theoretical prediction. At each step we can solve for $\mu_k$ from Equation (\ref{eq:unifsimp_partialkappa}) and we can calculate the gradient from Eqn.~(\ref{eq:unifsimp_partialkappa}). Note that here we have neither Murray nor material constraint, so a numerical projection is not required. The numerical optimal networks have highly heterogeneous conductances within each optimal network (Fig.~\ref{fig:Qsq_square_grid}A, B), but, as the theory predicted, the flow distribution agrees with the network with uniform conductance (Fig.~\ref{fig:Qsq_square_grid}C). 

\subsection{Optimal network for uniformizing flows with Murray constraint\label{sec:Murray_constaint}}

So far we have followed previous work\cite{bohn2007structure,katifori2010damage} by calculating all of our optimal networks under constraints on the total material. However both material investment and transport costs (i.e. dissipation) may contribute to the total cost of a particular network. We modify our cost function, $g$, to include both costs. In this case $g(\{p_k\},\{\kappa_{kl}\}) = \sum (a\kappa_{kl}(p_k-p_l)^2 + \kappa^\gamma_{kl})-K$ depends on both pressure and conductance, and the full mechanism for keeping $g$ constant during the gradient descent needs to be used. To calculate the optimal network by this method we need an explicit formula for $\lambda$. The details are somewhat involved, and we place them in \ref{App:Murray_lam}. 

Are optimal networks under Murray's constraint morphologically different from those only under material constraint? It is difficult to answer this question for general target functions because it requires us to understand how the constraint surface intersects with the target functions. However for target functions that only depend on flows such as the flow uniformity target function the scaling on conductances can give us additional information. Suppose we find an optimal network under the material constraint. We calculate the total material cost $K$ of this network. Then calculate the optimal network in which Murray's constraint is imposed with allowed total energy $K$ including both material costs and dissipation. Denote by $\kappa_{kl}$ the conductances in the network under Murray constraint, and by $\kappa'_{kl}$ the conductances in the optimal network under material constraint. If $a$ is sufficiently close to zero then the target function of Murray network will be lower or equal to that of material network. The reasoning is that although $\sum \kappa'^\gamma_{kl} + a\frac{Q_{kl}^2}{\kappa'_{kl}} = K$ does not hold, we can try to solve for a multiplicative scaling $\beta>0$ that satisfies $\sum (\beta \kappa'_{kl})^\gamma + a\frac{Q_{kl}^2}{\beta \kappa'_{kl}} = K$. Notice that $Q_{kl}$ does not change under the scaling for this class of networks, so the value of target function is unaffected by scaling conductances. Now if $a>0$ is small enough we expect to be able to find a solution $\beta$ and $\{\beta \kappa'_{kl}\}$ is an admissible network in the sense that it obeys the Murray constraint. Thus the optimal network obeying the Murray constraint must have equal or smaller target function value than the optimal network obeying only the material constraint. By reversing this argument we can see that the optimal networks for small enough $a>0$ actually agree with those with $a=0$. The question is how large $a$ has to be so that the Murray network is truly constrained by the total energy cost so that optimal networks under the Murray constraint and under the material constraint diverge. To approach the question we numerically obtained the optimal networks for uniform flow on the topology of capillary bed (Fig.~\ref{fig:Qsq_square_grid}A) with $0\leq a \leq 50$ and fixed total energy cost. The Murray networks look qualitatively similar to network with only material constraints (Fig.~\ref{fig:symnet_murray}A), and have the same values of target function the same as analytical lower bound (for a uniform conductance network) (Fig.~\ref{fig:symnet_murray}B). This result suggests that there could be a wide range of $a$ for which the Murray constraint and material constraints result in identical optimal networks. However the Murray constraint does have an effect on the relative strength of dissipation and material cost. We observe that increasing $a$ decreases material costs (Fig.~\ref{fig:symnet_murray}C). The trend is unintuitive since $a$ represents the relative costs of dissipation and material. We might therefore expect at larger values of $a$, the network would invest more in material to reduce dissipation. However if we study the curve of $\sum (\beta \kappa'_{kl})^\gamma + a\frac{Q_{kl}^2}{\beta \kappa'_{kl}}$ drawn as a function of $\beta$, the function is U-shaped and diverges if $\beta \to 0$ or if $\beta\to \infty$. When $a$ increases the total energy increases, and the network has to adjust itself to a low energy state. If the network is on the left side of the U this means increasing $\beta$, which increases the material cost to realize the constraint. In contrast when the network is on the right side of the curve, decreasing $\beta$ will be the only way to lower the total energy, which explains the trends depicted in Fig.~\ref{fig:symnet_murray}B. We will further dissect the role of $a$ in Section \ref{sec:zebrafish}.

\begin{figure}[ht!]

	\begin{center}
		\includegraphics[width=15 cm]{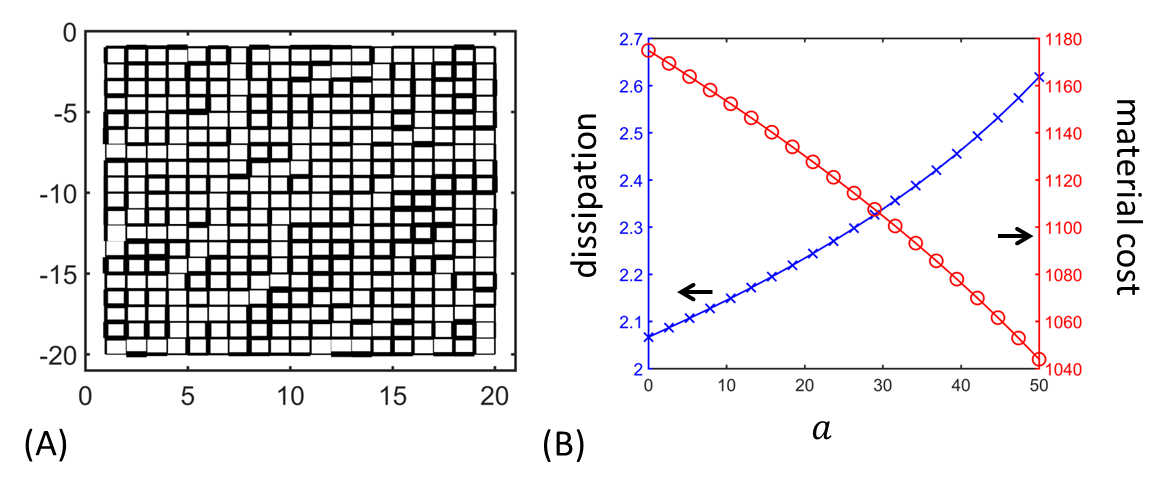}
		\caption{Optimal networks of functional $\frac{1}{2}\sum Q^2$ on a $20\times 20$ square grid under Murray constraint have the same flows as the analytic solution in Sec \ref{sec:Qsqanal}, but exhibit tradeoff between dissipation and material cost as $a$ increases. (A) For small $a$ optimal network with Murray constraint is equivalent to a network with material constraint. The network is constrained with $a = 36.8$, and the solution is selected from the best network visited during the gradient descent, with relative error in energy cost $<10^{-4}$, as in the following simulations. Widths show the relative conductances. (B) When $a$ is increased the dissipation in the network increases (blue crosses), while the material cost decreases (red circles). The simulations were carried out in the manner of numerical continuation, i.e. the simulation for each $a$ starts with the solution from previous $a$, and the simulation for $a=0$ starts with a random conductance configuration. All the networks have the same fixed total energy cost $K = 1174.9$.}
		\label{fig:symnet_murray}
\end{center}
\end{figure}

\section{Optimal networks on zebrafish embryo trunk vasculature \label{sec:zebrafish}}

Zebrafish are model organisms for studying vertebrate biology. In their embryonic state they are transparent, allowing the microvessels to be seen under the zebrafish's skin. Accordingly the embryonic zebrafish cardiovascular network is widely used to study vascular network growth and the effects of damage on the network\cite{lawson2002vivo,lieschke2007animal,chico2008modeling,walcott2014zebrafish,isogai2001vascular}. Blood flows into the trunk of the zebrafish through the dorsal aorta and then passes into minute vessels called intersegmental (Se) vessels. Blood then returns to the heart via the cardinal vein. These vessels are arranged just like rungs (Se) and parallels (cardinal vein and dorsal aorta) of a ladder (Fig.~\ref{fig:zebrafish_diss_topo}A). Most gas exchange in the network is assumed to occur in the Se vessels. As the zebrafish develops further minute vessels form between the Se vessels, converting the trunk into a dense reticulated network\cite{isogai2001vascular}. We focus on the mechanisms underlying flow distribution in the main fine vessels (Fig.~\ref{fig:zebrafish_diss_topo}A). Our previous study of the zebrafish microvasculature\cite{chang2015optimal} showed that if each vessel has the same radius then most red blood cells would return to the heart via the highest conductance path, i.e. along the closest Se vessel to the heart, which effectively acts as a short circuit for the network. Our analysis also revealed tradeoffs between preventing short circuits and increasing the dissipation within the network; that is, more flow would pass through distant Se vessels if the conductance of distant Se vessels is increased. But this distribution of conductances has higher dissipation than a network in which all Se vessels have the same conductance. Moreover, although the observed distribution of conductances does not create exactly uniform flows across all Se vessels, creating more uniform distributions of flow would further increase the dissipation within the network. The optimization method described in this paper arose as a way to create a mathematically formal version of the problem: with a given total energy available, how uniformly can flows be divided between intersegmental vessels, and how close is the real embryonic zebrafish network to this constrained optimum?

\begin{figure}[H]

	\begin{center}
		\includegraphics[width=12 cm]{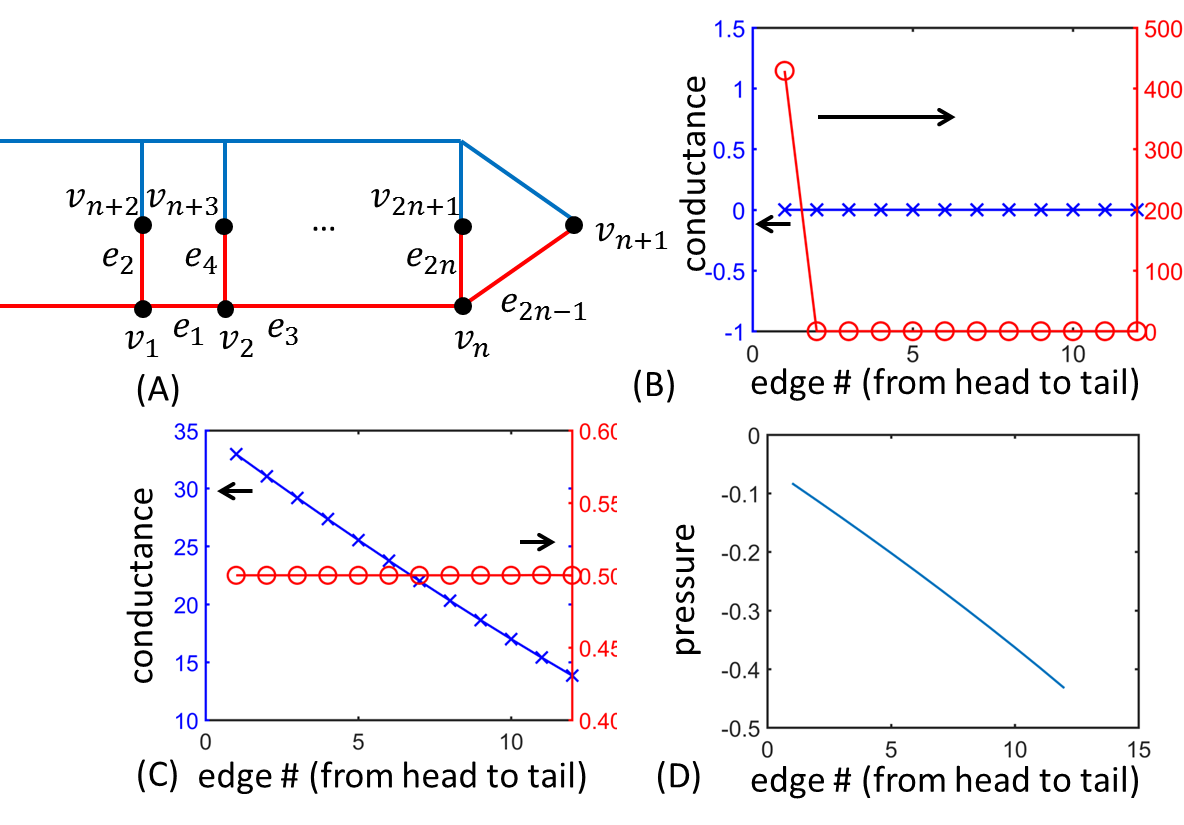}
\caption{Minimal dissipative networks for zebrafish trunk vasculature. (A) The zebrafish trunk vasculature can be simplified into a ladder network with aorta part (red) and the vein part (blue). The edges $e_1,e_3,...,e_{2n-1}$ are aorta segments and $e_2,e_4,...,e_{2n}$ are capillaries. We use $n=12$ in all the following calculations on zebrafish network. (B) The optimal dissipative network with $\gamma = \frac{1}{2}$ and fixed inflow does not correctly describe the zebrafish trunk network since all the conductances are concentrated on the first capillary (red circle), and the whole aorta is deleted (blue cross). In this calculation we imposed a fixed inflow on $v_1$ and fixed zero pressure on $v_{n+1},...,v_{2n+1}$. We started with $\kappa = 20$ for aorta segments and $\kappa=1$ for capillaries to reflect the difference in radii in real zebrafish. This initial condition is used for all the following simulations. (C) The optimal dissipative network with $\gamma=\frac{1}{2}$ and fixed outflows has a tapering aorta (blue cross) and capillaries with the same conductances (red circle). We imposed zero pressure on $v_1$ and fixed outflows on $v_{n+1},...,v_{2n+1}$ with $v_{n+1}$ taking half of the total outflow (i.e. $\frac{1}{2}F$) and $v_{n+2},...,v_{2n+1}$ evenly dividing the other half of $F$. (D) However the pressures on the ends of capillaries are decreasing to maintain uniform flows among capillaries, which is not physical due to the aorta-vein symmetry.}

\label{fig:zebrafish_diss_topo}
\end{center}

\end{figure}

Since the zebrafish trunk network is symmetric we can just consider half of the network consisting of the aorta and intersegmental arteries, designated by vertices $v_1,...,v_{2n+1}$ and edges $e_1,...,e_{2n}$ with $n$ being the number of Se vessels (Fig.~\ref{fig:zebrafish_diss_topo}A). Due to the symmetry of the zebrafish trunk vasculature we fix the pressures at $v_{n+1},...,v_{2n+1}$. We assume the heart pumps a constant volume of blood into the trunk in every time interval so we apply a fixed inflow, $F$, boundary condition on $v_1$. First we show how far the network is from minimizing dissipation. If we assign a cost function based only on the total material in the network (i.e. set $a=0$ and $\gamma=\frac{1}{2}$ in Eqn.~(\ref{eq:Thetadefinition})), then minimizing dissipation eliminates all but the first Se vessel (Fig.~\ref{fig:zebrafish_diss_topo}B). Conversely if we instead impose uniform flow at each of the vertices  $v_{n+2},...,v_{2n+1}$ and seek a distribution of conductances that minimizes dissipation, although we see a more realistic distribution of conductances (identical conductances in each Se vessel and tapering aorta (Fig.~\ref{fig:zebrafish_diss_topo}C)), in this optimal network the pressures where the Se vessels meet the cardinal vein decrease with distance from the heart (Fig.~\ref{fig:zebrafish_diss_topo}D), so that blood flows away from the heart within the cardinal vein which is unphysical.

We then explore an alternate organizing principle. Specifically we make uniform flow within Se vessels as our target function. Consider the functional
\begin{equation}
f(\{p_k\},\{\kappa_{kl}\}) = \sum_{i=1}^n \frac{1}{2}(Q_{2i} -\bar{Q})^2,
\end{equation}
where $\bar{Q}$ is a predetermined flow for all the capillaries (in the following arguments edge-defined quantities such as $Q_i$ are indexed with the edges, and vertex-defined quantities such as $p_i$ are indexed with the vertices). Using this indexing scheme, the function to be optimized becomes:
\begin{align}
\Theta = & \sum_{i=1}^n \frac{1}{2}\kappa^2_{2i} p^2_i - \sum_{i=1}^n \bar{Q}\kappa_{2i}p_i - \sum_{i=2}^{n-1} \mu_i[\kappa_{2i-3}(p_i-p_{i-1}) + \kappa_{2i-1}(p_i-p_{i+1}) + p_i\kappa_{2i}] \nonumber \\
 & - \mu_1[\kappa_1 (p_1-p_2) + p_1\kappa_2 - F] - \mu_n[\kappa_{2n-3}(p_n-p_{n-1}) + p_n \kappa_{2n-1} + \kappa_{2n}p_n].\label{functional:zebrafish_QmQbarsq}
\end{align}
Just as in Section \ref{sec:Qsqanal} we do not need to introduce a Lagrange multiplier enforcing the material constraint because the target function only depends on flows, and we can scale all conductances to realize any material constraint without affecting the target function. We put the details of the calculation in \ref{App:zebrafish_grad}. Instead of concentrating all the materials on the first capillary or tapering the aorta, the uniform flow network has constant conductance along the aorta and conductances on the Se vessels that increase exponentially with distance from the heart (Fig.~\ref{fig:zebrafish_Qbar_kappa}A). Previously\cite{chang2015optimal} we showed that if each Se vessel is assigned the same conductance, then blood flows will decrease exponentially with the index of the Se vessel. To counter this effect and to achieve uniform flow the conductance of Se vessels has to increase from head to tail. Indeed the optimal distribution of conductances matches closely to the experimental data we measured\cite{chang2015optimal} (Fig.~\ref{fig:zebrafish_Qbar_kappa}B), further suggesting that uniformity might be prioritized over dissipation within zebrafish cardiovascular network.

\begin{figure}[h!]

	\begin{center}
		\includegraphics[width=12 cm]{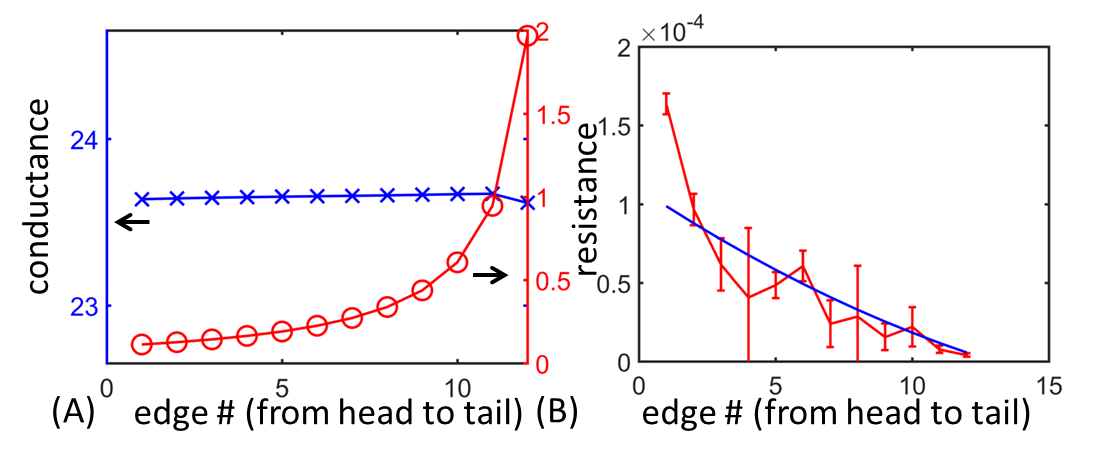}
\caption{The optimal distribution of material for achiving uniform flows (\ref{functional:zebrafish_QmQbarsq}). (A) The optimal network dictates a constant conductance on aorta segments (blue cross) but assigns conductances to Se vessels that increase exponentially from head to tell (red circle). We scale the conductances such that $\sum \kappa^\frac{1}{2}$ remains the same for comparison with minimal dissipative networks. (B) The predicted hydraulic resistance (blue dashed curve) agrees well with experimentally measured data (red curve, with 95\% confidence intervals). The data is obtained from our previous work\cite{chang2015optimal} under the assumption that the volume fraction of the red blood cells is $0.45$. The theoretical prediction is normalized by the mean of the data.}

\label{fig:zebrafish_Qbar_kappa}
\end{center}

\end{figure}

The real zebrafish network agrees well with the optimal set of conductances predicted for a network that uniformizes fluxes across Se vessels. But the agreement is not exact. Is the difference between the two optimal and real networks evidence that the real network has other constraints or target functions that are not modeled by Equation (\ref{functional:zebrafish_QmQbarsq})? When given two potential target functions or constraints that may explain the measured geometry of a real transport network, our optimization method provides tools to measure the relative weight the network gives to the two principles. For the zebrafish network, we perform network optimization using the Murray constraint, varying the parameter $a$ to see the extent to which material or transport costs influence the network organization. The gradient descent method with Murray constraint follows \ref{App:Murray_lam} with the target function (and therefore the formula for $\mu$) modified. Specifically $\chi$ now becomes
\begin{equation}
 \chi_{kl} = (\kappa_{kl}(p_k-p_l)-\bar{Q})(p_k-p_l)I_{kl} - \nabla (D^{-1}\zeta)_{kl} \nabla p_{kl} \qquad (k,l)\in \mathcal{E}
\label{eq:zebrafish_chi}
\end{equation}
where $I_{kl} = 1$ if and only if the edge $kl$ is an intersegmental vessel ($\mathcal{E} \doteq \{(k,l): \langle k,l\rangle =1, k<l\}$) and
\begin{equation}
\zeta_k \doteq \left\{\begin{array}{ll}
\sum_{l,\langle k,l\rangle = 1} (\kappa_{kl}(p_k-p_l)-\bar{Q})I_{kl}\kappa_{kl} & k\notin \Vd\\
0 & k\in \Vd\\
\end{array}\right.
\label{eq:zebrafish_zeta}
\end{equation}
Finally once $\lambda$ has been solved for, the expression of $\mu_k$ is calculated from
\begin{equation}
\mu = 2a\lambda p + D^{-1}\zeta
\end{equation}
(For complete derivation see \ref{App:Murray_lam_zebrafish}). Based on our analysis (in Section \ref{sec:Murray_constaint}) of uniform partitioning of flows in networks with the Murray constraint, we expect that the optimal zebrafish network will be essentially independent of $a$ over some finite interval of $a$ values, starting at $0$. Indeed we find that for small $a$ the target function remains vanishing and the dissipation increases as $a$ increases up to a critical value. However, the arguments given in Section \ref{sec:Murray_constaint} are silent on how the network changes as $a$ is increased, in particular what happens once $a$ exceeds the critical value, once $a$ exceeds the threshold where it is no longer possible to rescale the conductances in a network that obeys a material constraint into a network that obeys the Murray constraint. We find that a critical value of $a_c = 33.3$ the network undergoes a phase transition where the target function switches from constant to monotonic increasing and the dissipation decreases (Fig.~\ref{fig:zebrafish_murray}A). At the phase transition the conductances of intersegmental vessels transition from solution shown in Fig.~\ref{fig:zebrafish_Qbar_kappa}B to becoming non-monotonic with the conductance increasing between vessels near the head and then decreasing at the tail (Fig.~\ref{fig:zebrafish_murray}B). Above the critical value of $a$, the optimal network no longer keep flows uniform between intersegmental vessels (Fig.~\ref{fig:zebrafish_murray}C). Put another way, as the parameter $a$ is changed, rather than smoothly interpolating between networks that optimize uniformity and networks that optimize transport, the network optimizes uniformity over a large interval of values of $a$, and then shifts suddenly to a network that is far from realizing a uniform distribution of fluxes.

\begin{figure}[h]

	\begin{center}
		\includegraphics[width=15 cm]{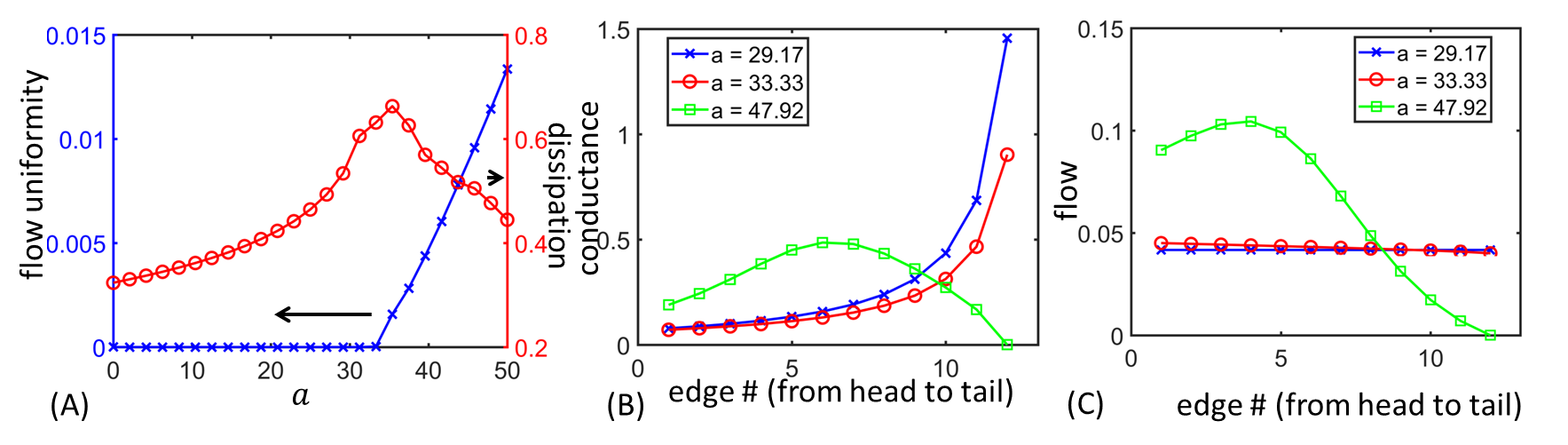}
\caption{Optimal networks of $\frac{1}{2}\sum_i (Q_{2i}-\bar{Q})^2$ exhibit an apparent phase transition under the Murray constraint with varying cost-of-dissipation. (A) The target function is vanishing for small $a$ until $a_c= 33.3$ where a phase transition occurs and the value of target function suddenly increases (blue cross). The dissipation (red circle) increases when $a<a_c$ due to a similar mechanism as in Fig.~\ref{fig:symnet_murray}B, but has a sharp decrease right after the critical value $a_c$. Here we adopted numerical continuation as in Fig.~\ref{fig:symnet_murray}B, but when a local minimum around previous initial condition does not satisfy Murray constraint the initial configuration at $a=0$ is reused as the initial conductances. The minimal value for the total energy cost upon scaling of conductances is used whenever the Murray constraint cannot be maintained. The Murray energy $K$ is maintained to be 70.43 in all simulations, justifying this projection method for stiff constraint. The total energy cost is fixed to that of initial configuration (with uniform conductances in capillary being $1$ and those in aorta being $20$) when $a=1$. The solution is selected from the best network visited during the gradient descent, with relative error in energy cost $<10^{-4}$ (B) The conductances of capillaries change qualitatively after the phase transition. The morphology resembles unconstrained network (Fig.~\ref{fig:zebrafish_Qbar_kappa}A) before the phase transition (blue cross and red circle), but changes qualitatively afterwards (green square). (C) The optimal value in target function is reached before the phase transition (blue cross and red circle) but flows decrease from head to tail afterwards (green square).}

\label{fig:zebrafish_murray}
\end{center}

\end{figure}

\section{Discussion and Conclusion}

In this work we proposed an algorithm that is able to find locally optimal networks for general target functions under general constraints. We tested that our algorithm is able to reproduce networks that agree with previously calculated optimal transport networks. Motivated by our previous work on zebrafish microvasculature\cite{chang2015optimal}, we then studied optimal networks that uniformize network flow and derived an analytical result confirmed by the numerical solutions. To study the tradeoffs between different target functions for a network we introduced a constraint that accounts for both the material cost and dissipation. Finally we applied our algorithm to the zebrafish trunk vasculature and showed that the numerical optimal network agrees with the experimental data. Moreover our results expose a phase transition that occurs as the relative size of transport and material costs is increased. Surprisingly, optimal networks do not continuously interpolate between optimizing uniformity and optimizing dissipation, but instead are initially invariant under changes in the cost of dissipation, and then undergo a sudden phase transition-like reconfiguration when this cost exceeds a certain threshold.

Although this result would need to be replicated for other combination of target functions, it offers a surprising biological insight; namely, the departure of real zebrafish networks from the optimum for creating uniform distributions of fluxes cannot be explained from the point of view of the network needing to balance tradeoffs between multiple target functions, and is therefore more likely due to another cause; for example variability (or noise) during vessel formation. More generally adherence to a single target function supports a continued focus on single target functions when studying biological networks, since no two functions will likely shape the network simultaneously.

Our algorithm treats the conductances of all edges as independent variables, so the number of degrees of freedom over which optimization is performed is the number of vessels. But the number of vessels in real biological networks may be so large as to defeat direct application of the algorithm. For example, in the mouse brain vascular network there are $\sim 10^4$ capillaries in a volume of $2 \; \textrm{mm}^3$ \cite{blinder2013cortical}. More degrees of freedom will also lead to a multiplication of local optima. While parallelization and coupling to global optimization methods for navigation rough landscapes (e.g. simulated annealing) could be potential solutions, another approach is to treat the brain as a multiscale network. Large vessels play different roles from small vessels (such as capillaries). This property may be exloited by numerical methods that treat different scales in different ways.

There are many other biological relevant functions to which our algorithm could be applied, for example damage resistance\cite{katifori2010damage} and mixing\cite{roper2013nuclear}. Moreover, our model of oxygen perfusion (which we assume to be uniform, so long as fluxes are uniform between fine vessels) is unlikely to be quantitatively correct for more complex networks. Specifically red blood cells will have lower oxygenation levels the more capillaries they travel through. The history of red blood cell passage through the network will therefore influence their oxygenation. 

Most optimization problems in this work are constrained either by material or total energy, and it is not clear whether imposing network cost limits as a penalty function rather than as a constraint will give the same result or not. In Murray's original paper the Murray's law was derived by minimizing the total energy formed as a sum of material and transport costs\cite{murray1926physiological,sherman1981connecting}. However recent works on minimal dissipation networks impose the material cost as a constraint and minimize dissipation under this constraint. The two approaches carry different physical meanings, and it is not clear which approach is a better model for real biological systems, or whether, indeed, they produce equivalent networks. We are currently studying the conditions under which the two problems are equivalent, i.e. produce equivalent classes of optimal networks\cite{chang2017minimal}.

In conclusion we proposed a gradient descent algorithm that finds optimal networks with general target functions and constraints. We create this algorithm to reveal the biological organizing principles of microvascular networks. The recent explosion in data streams for microvasculature geometry and flow\cite{chaigneau2003two,blinder2013cortical}, has created an unmet need for quantitative tools for testing hypotheses on the optimization principles underlying real transport networks. As our zebrafish study shows, our algorithm allows comparison between biological networks and optimal networks achieving different biological functions. While further work will be needed to resolve computational challenges and make rigorous mathematical formulation, our work provides a way to test hypothetical optimal trategies for microvasculature organization, with long term use when understanding microvascular damage, defects and recovery.

\section{Acknowledgments}

This research was funded by grants from the NSF (under grant DMS-1351860). MR. SSC was also supported by the National Institutes of Health, under a Ruth L. Kirschstein National Research Service Award (T32-GM008185). The contents of this paper are solely the responsibility of the authors and do not necessarily represent the official views of the NIH. MR also thanks Eleni Katifori and Karen Alim for useful discussions, and the American Institute of Mathematics for hosting him during one part of the development of this paper.

\appendix\section{Solvability of $\{\mu_k\}$\label{App:mu_solve}}

Here we prove that $\{\mu_k\}$ in Equation (\ref{eq:muequation}) are solvable under a general configuration of flow (i.e. Neumann) and pressure (i.e. Dirichlet) boundary conditions (BCs). We assume that $\kappa_{kl} >0 \; \forall \langle k,l\rangle = 1$ (since $\kappa_{kl} = 0$ is the same as $\langle k,l\rangle = 0$) and that the network is connected. It suffices to show that the matrix $D$ is invertible. However this is the same matrix in the linear system for solving $\{p_k\}$ with the specified BCs, so we only have to show that there exists a unique flow given any flow and pressure BCs, which is a well-known\cite{LP:book}. However since our derivation makes use of multiple invertibility results for different matrices $D, D\ub{2}$ and so on, we provide a proof in order to highlight under what conditions invertibility is allowed. The problem is equivalent to showing that

\begin{equation}
Dp = 0 \Rightarrow p =0.
\label{eq:no_flow_pressure}
\end{equation}

\noindent The solution $p$ for Eqn.~(\ref{eq:no_flow_pressure}) corresponds to a network where we do not have any flows into the system except possibly at nodes with pressure BCs, denoted by $\Vd$. The goal is to show that $p_k = 0 \; \forall k$. Suppose for contradiction that $\exists i\notin \Vd$ s.t. $p_i \neq 0$ (since we already have $p_j = 0 \; \forall j\in \Vd$). Then we would have $Q_{kl} \neq 0$ for some $\langle k,l\rangle =1$ since the network is connected, and WLOG let $Q_{kl}>0$. Now we can trace this flow throughout the network in the following procedure:

\begin{enumerate}
\item Given that $Q_{k_{n-1} k_n} >0$ first check if $k_n \in \Vd$, and stop if this is the case.
\item Consider all nodes $l$ s.t. $\langle k_n,l\rangle = 1$. According to Kirchhoff's first law there must be an $l$ s.t. $Q_{k_n l} >0$. Since the network is finite we can pick e.g. the smallest $l$ satisfying these conditions and let $k_{n+1} = l$.
\item Repeat the procedure until $k_N \in \Vd$ for some $N$ and stop. 
\end{enumerate}

If we start with $k_1 = k, k_2 = l$ we can initiate the process since the first condition is satisfied. This procedure has to stop eventually because the network is finite and that $k_1,...,k_n$ are all distinct for any given $n>1$. To see this suppose $k_n = k_m$ with $m>n$. Then we would have $p_n>p_{n+1}>\cdots > p_m = p_n$, a contradiction. Thus we would end up with a chain of distinct nodes $k_1,k_2,...,k_N$ with $\langle k_n,k_{n+1}\rangle =1, Q_{k_n k_{n+1}} > 0 \; \forall n=1,...,N-1$, and $N\in \Vd$. Now we repeat the same procedure just with $k'_1 = l, k'_2 = k$ to trace the flows upstream, and we would end up with another chain $k'_1,k'_2,...,k'_{N'}$ with $\langle k'_n,k'_{n+1}\rangle =1, Q_{k'_n k'_{n+1}} < 0 \; \forall n=1,...,N'-1$, and $N' \in \Vd$. Notice that there is no repetition in the set $\{k_1,...,k_N,k'_1,...,k'_{N'}\}$  since $k_n = k'_m$ would lead to the same contradiction since pressures must be ordered.

\section{Explicit formula for $\lambda$ for uniform flow networks with Murray constraint\label{App:Murray_lam}}

We introduce several notations to be used later. Suppose $\{b_{ij}\}$ is a set of quantities defined on the edges of the network. For any real constant $c$ we define the matrix for the graph Laplacian with specified boundary conditions as 

\begin{equation}
M_b\ub{c} = \left\{\begin{array}{llll}
\sum_{l,\langle k,l\rangle =1} b_{kl} & k=l, k\notin \Vd\\
-b_{kl} & \langle k,l\rangle =1\\
c & k=l, k\in \Vd\\
0 & otherwise
\end{array}\right..
\end{equation}
We also abbreviate $M_b = M_b\ub{1}$. In the notation of Equations (\ref{eq:unifsimp_kappasqmatrix}) $D = M_\kappa$ and $D\ub{2} = M_{\kappa^2}$. For a quantity $v$ that is defined on the vertices of the network (such as pressure) we define the graph difference vector $\nabla v \in \R^E$ as

\begin{equation}
\nabla v_{kl} = v_k-v_l \qquad (k,l)\in \E,
\end{equation}
where $\mathcal{E}$ denotes the set of ordered pairs of edges so that each edge only appear once in $\mathcal{E}$. Now we can derive the formula for $\lambda$: $\delta\kappa$ is given by the explicit formula. From $\frac{\partial\Theta}{\partial p_k} = 0$ we obtain $\mu = D^{-1}D\ub{2} p + 2\lambda a p$ (recall here we have $f = \sum_{k>l, \langle k,l\rangle =1} \frac{1}{2}\kappa^2_{kl}(p_k-p_l)^2, g = \sum_{k>l, \langle k,l\rangle =1} a\kappa_{kl}(p_k-p_l)^2 + \kappa^\gamma_{kl} - K^\gamma$), and so:

\begin{equation}
\frac{\partial\Theta}{\partial\kappa_{kl}} = \lambda[\gamma \kappa^{\gamma-1}_{kl} - a(\nabla p_{kl})^2] + \kappa_{kl} (\nabla p_{kl})^2  - \nabla(D^{-1}D\ub{2}p)_{kl}\nabla p_{kl}.
\end{equation}
 
\noindent We determine $\lambda$ from the variational:

\begin{align}
 0  =  dg &  = \sum_{k>l,\langle k,l\rangle =1} \gamma \kappa_{kl}^{\gamma-1} \delta \kappa_{kl} + a\delta\kappa_{kl} \nabla p_{kl}^2 + 2a \kappa_{kl} \nabla \delta p_{kl} \nabla p_{kl} \nonumber\\
 & =  \sum_{k>l,\langle k,l\rangle =1} -\alpha (\gamma\kappa_{kl}^{\gamma-1}+a\nabla p_{kl}^2)\Big\{\lambda[\gamma \kappa^{\gamma-1}_{kl} - a\nabla p_{kl}^2] \nonumber\\
 & + \kappa_{kl} \nabla p_{kl}^2  - \nabla(D^{-1}D\ub{2}p)_{kl}\nabla p_{kl}\Big\} + 2a \kappa_{kl}\nabla \delta p_{kl} \nabla p_{kl}.
\end{align}
This formula depends on $\delta p$; the change in $p$ produced by the change $\kappa\mapsto \kappa+\delta\kappa$. If we assume $p_i = 0 \; \forall i\in \Vd$ we can write Equation (\ref{eq:Kirchfirstvar}) in matrix form as

\begin{equation}
M_{\delta \kappa}p + D\delta p = 0
\label{eq:Kirchfirstvar_matrix}
\end{equation}

\noindent so

\begin{equation}
\delta p = -D^{-1} M_{\delta \kappa} p.
\end{equation}
(Equation (\ref{eq:Kirchfirstvar_matrix}) can be modified by adding a non-zero vector on the right hand side, if inhomogeneous pressure boundary conditions are applied.) Thus if we define auxiliary variables: $\beta\doteq \gamma \kappa^{\gamma-1} - a\nabla p^2, \chi \doteq \kappa \nabla p^2 - \nabla(D^{-1}D\ub{2} p)\nabla p$, so that $\delta \kappa = -\alpha(\lambda\beta+\chi)$, then:

\begin{align}
0 & = -\alpha \Big\{\lambda \sum_{k>l,\langle k,l\rangle =1}(\gamma\kappa_{kl}^{\gamma-1} + a\nabla p_{kl}^2)\beta_{kl} + \sum_{k>l,\langle k,l\rangle =1}(\gamma \kappa_{kl}^{\gamma-1} + a\nabla p_{kl}^2)\chi_{kl}\Big\} \nonumber \\
& - 2a \sum_{k>l,\langle k,l\rangle =1}\kappa_{kl}\nabla p_{kl} \nabla (D^{-1} M_{-\alpha\{\lambda\beta + \chi\}} p)_{kl}, \nonumber \\
0 & =\lambda \sum_{k>l,\langle k,l\rangle =1} \gamma^2 \kappa_{kl}^{2\gamma-2} - a^2\nabla p_{kl}^4 - 2a\kappa_{kl}\nabla p_{kl} \nabla(D^{-1}M\ub{0}_{\beta} p)_{kl} \nonumber \\
& +\sum_{k>l,\langle k,l\rangle =1} (\gamma\kappa_{kl}^{\gamma-1} + a\nabla p_{kl}^2)\chi_{kl} - 2a\kappa_{kl}\nabla p_{kl} \nabla (D^{-1}M\ub{-\frac{1}{\alpha}}_{\chi}p)_{kl}.
\end{align}

\noindent Finally we can write down the formula for $\lambda$ as

\begin{equation}
\lambda = \frac{-\sum_{k>l,\langle k,l\rangle =1} (\gamma\kappa_{kl}^{\gamma-1} + a\nabla p_{kl}^2)\chi_{kl} - 2a \kappa_{kl}\nabla p_{kl} \nabla (D^{-1}M\ub{-\frac{1}{\alpha}}_{\chi}p)_{kl}}{\sum_{k>l,\langle k,l\rangle =1} \gamma^2 \kappa_{kl}^{2\gamma-2} - a^2 \nabla p_{kl}^4 - 2a \kappa_{kl}\nabla p_{kl} \nabla(D^{-1}M\ub{0}_{\beta} p)_{kl}}.
\label{eq:Murray_lambda}
\end{equation}

The value of $\lambda$ in Eqn.~(\ref{eq:Murray_lambda}) ensures that $g$ remains constant up to $O(\delta \kappa_{kl})$ terms. However, we must also adjust $\{\kappa_{kl}\}$ at each step to exactly maintain the constraint following the method given in Section \ref{sec:algorithm}. In previous applications since $g$ was a function of $\kappa$ alone this additional projection step did not require perturbation of pressures. Now both the change in $\kappa_{kl}$ and the change in flow must be considered when adjusting conductances. We calculate here the additional terms created by involvement of pressures. To project along the constraint surface normal we need to calculate the normal vector:

\begin{align}
n_{kl}& = \frac{\partial}{\partial\kappa_{kl}}\Big\{\sum_{i>j,\langle i,j\rangle} \big(\kappa^\gamma_{ij} + a\kappa_{ij}(p_i-p_j)^2 \big) - K^\gamma\Big\} \nonumber \\
& = \gamma \kappa^{\gamma-1}_{kl} + a(p_k-p_l)^2 + \sum_{\langle i,j\rangle,i>j} 2a\kappa_{ij}(\frac{\partial p_i}{\partial\kappa_{kl}} - \frac{\partial p_j}{\partial\kappa_{kl}})(p_i-p_j).
\end{align}

\noindent To obtain $\frac{\partial p_i}{\partial\kappa_{kl}}$ we differentiate Kirchhoff's first law with respect to $\kappa_{kl}$:

\begin{equation}
\sum_j \kappa_{ij} (\frac{\partial p_i}{\partial\kappa_{kl}} - \frac{\partial p_j}{\partial\kappa_{kl}}) + (\delta_{ik}\delta_{jl}-\delta_{il}\delta_{jk})(p_i-p_j) = 0
\end{equation}

\noindent or:

\begin{equation}
\sum_j \kappa_{ij} (\frac{\partial p_i}{\partial\kappa_{kl}}-\frac{\partial p_j}{\partial\kappa_{kl}}) = -(p_k-p_l)(\delta_{il}+\delta_{ik}).
\end{equation}

\noindent Notice that $\frac{\partial p_i}{\partial\kappa_{kl}} = 0\; \forall i\in \Vd$ since these $p_i$ are fixed by the boundary conditions. Then we can solve for $\frac{\partial p_i}{\partial \kappa_{kl}}, 1\leq i \leq V$ by solving the linear system (solvability was discussed in \ref{App:mu_solve}) and calculate the normal vector.\\

\section{Gradient descent method for zebrafish trunk network uniformizing flows in intersegmental vessels \label{App:zebrafish_grad}}

For performing gradient descent method for zebrafish trunk network uniformizing flows in Se vessels we calculate the partial derivatives of $\Theta$:

\begin{equation}
\frac{\partial\Theta}{\partial p_i} = \left\{\begin{array}{lll}
\kappa^2_{2i} p_i - \bar{Q}\kappa_{2i} - (\kappa_{2i-1} + \kappa_{2i-3} + \kappa_{2i})\mu_i + \kappa_{2i-1}\mu_{i+1} + \kappa_{2i-3}\mu_{i-1} & i\neq 1,n\\
\kappa^2_2 p_1 - \bar{Q}\kappa_2 - (\kappa_1+\kappa_2)\mu_1 + \mu_2\kappa_1 & i=1\\
\kappa^2_{2n} p_n -\bar{Q}\kappa_{2n} - (\kappa_{2n-3} + \kappa_{2n-1} + \kappa_{2n}) \mu_n + \kappa_{2n-3}\mu_{n-1} & i=n
\end{array}\right..
\label{eq:zebrafish_partialp}
\end{equation}

\begin{equation}
\frac{\partial\Theta}{\partial\kappa_i} = \left\{\begin{array}{lll}
\kappa_i p^2_{i/2} - \bar{Q}p_{i/2} - \mu_{i/2}p_{i/2} & i|2 = 0\\
-(\mu_{\frac{i+1}{2}} - \mu_{\frac{i+3}{2}})(p_{\frac{i+1}{2}} - p_{\frac{i+3}{2}}) & i|2 =1, i\neq 2n-1\\
-\mu_n p_n & i=2n-1\\
\end{array}\right..
\label{eq:zebrafish_partialkappa}
\end{equation}

Then we impose the physical BCs, i.e. fixed inflow into the network and zero pressure on the ends of the main aorta and the capillaries, and perform gradient descent to find the optimal network.

\section{Explicit formula for $\lambda$ for uniform flow networks with Murray constraint on zebrafish trunk vascular network \label{App:Murray_lam_zebrafish}}

Here we carry out the calculation of $\{\mu_k\}, \chi$ for $\lambda$ calculation on zebrafish trunk vascular network topology, following \ref{App:Murray_lam}. The only difference lies in the target function:

\begin{equation}
f = \sum_{(k,l)\in \mathcal{E}} \frac{1}{2}(\kappa_{kl}(p_k-p_l)-\bar{Q})^2 I_{kl}
\end{equation}
where $\mathcal{E} = \{(k,l): \langle k,l\rangle =1, k<l\}$ under the zebrafish trunk topology and our index convention (Fig.~\ref{fig:zebrafish_diss_topo}A), and $I$ is defined as in Equation (\ref{eq:zebrafish_chi}). Again from $\frac{\partial\Theta}{\partial p_k} = 0$ we get

\begin{equation}
\mu = 2a\lambda p + D^{-1}\zeta
\label{eq:mu_zebrafish_app}
\end{equation}
where $\zeta$ is defined as in Equation (\ref{eq:zebrafish_zeta}). Then the gradient of $\Theta$ can be calculated as

\begin{equation}
\frac{\partial\Theta}{\partial\kappa_{kl}} = (\kappa_{kl}(p_k-p_l)-\bar{Q})(p_k-p_l)I_{kl} - a\lambda (\nabla p^2)_{kl} + \lambda \gamma \kappa^{\gamma-1}_{kl} - \nabla (D^{-1}\zeta)_{kl} \nabla p_{kl} \doteq \lambda\beta_{kl} + \chi_{kl} \qquad \forall (k,l)\in E
\end{equation}
where $\beta_{kl} = \gamma\kappa^{\gamma-1}_{kl}$ as in \ref{App:Murray_lam}, but $\chi_{kl} = (\kappa_{kl}(p_k-p_l)-\bar{Q})(p_k-p_l)I_{kl} - \nabla (D^{-1}\zeta)_{kl} \nabla p_{kl}$ is different. Notice that if we set $\bar{Q} = 0, I_{kl} = 1 \; \forall \langle k,l\rangle =1$ then $f$ is the same as in \ref{App:Murray_lam} and the expression of $\chi$ agrees with that in \ref{App:Murray_lam}. Since the expression of $\beta$ does not change we can simply plug $\chi$ into Equation (\ref{eq:Murray_lambda}) to obtain $\lambda$, and use Equation (\ref{eq:mu_zebrafish_app}) to obtain $\{\mu_k\}$ for the gradient descent.

\end{document}